\documentclass[11pt]{article} %

\usepackage[backend=bibtex,style=numeric-comp,natbib=true,maxalphanames=10,maxnames=3,sortcites=true]{biblatex} %

\usepackage[margin=1in]{geometry}

\usepackage{microtype}          %
\usepackage[T1]{fontenc}
\usepackage[full]{textcomp}

\usepackage{article} %
\usepackage{math} %
\usepackage{algo} %
\usepackage{amsfonts}           %
\usepackage{wrapfig} %
\usepackage{placeins} %
\usepackage{graphicx} %
\usepackage{layout} %
\usepackage{setspace} %
\usepackage{mathtools} %
\usepackage{grffile} %
\usepackage{pdfpages} %
\usepackage{mdframed} %
\usepackage{flows}    %
\usepackage{booktabs} %
\usepackage{tabulary} %

\newcommand{\defterm}[1]{{\boldmath\normalfont \bfseries #1}}

\makeatletter
\g@addto@macro\bfseries{\boldmath}
\makeatother

\bibliography{references}

\newcommand{\graybar}{{\color{gray}\hrulefill}}

\newcommand{\SMtime}{\fpair{\operatorname{SM}}} %
\newcommand{\SMItime}{\fpair{\operatorname{SMI}}} %
\newcommand{\SFMtime}{\fpair{\operatorname{SFM}}}

\newcommand{\emstrut}[2]{\vrule height #1em depth #2em width 0pt }

\providecommand{\barV}{\overline{V}} %
\providecommand{\barU}{\overline{U}} %
\providecommand{\elements}{V \cup E}
\providecommand{\elemconn}{\kappa'}

\title{Isolating~Cuts, (Bi-)Submodularity, and
  Faster~Algorithms~for~Global~Connectivity~Problems}

\author{
Chandra Chekuri\thanks{Dept.\ of Computer Science, Univ.\ of Illinois, Urbana-Champaign, Urbana,
  IL 61801. {\tt chekuri@illinois.edu}. Supported in part by NSF grants
CCF-1910149 and CCF-1907937.}
\and
Kent Quanrud\thanks{Dept.\ of Computer Science, Purdue University,
  West Lafayette, IN 47909. {\tt krq@purdue.edu}.}
}

\begin{document}

\maketitle

\begin{abstract}
  \citet{li-panigrahi}, in recent work, obtained the first
  deterministic algorithm for the global minimum cut of a weighted
  undirected graph that runs in time $o(mn)$.  They introduced an
  elegant and powerful technique to find \emph{isolating cuts} for a
  terminal set in a graph via a small number of $s$-$t$ minimum cut
  computations.

  In this paper we generalize their isolating cut approach to the
  abstract setting of symmetric bisubmodular functions (which also
  capture symmetric submodular functions). Our generalization to
  bisubmodularity is motivated by applications to element connectivity
  and vertex connectivity. Utilizing the general framework and other
  ideas we obtain significantly faster randomized algorithms for
  computing global (and subset) connectivity in a number of settings
  including hypergraphs, element connectivity and vertex connectivity
  in graphs, and for symmetric submodular functions.
\end{abstract}

\section{Introduction}
\labelsection{intro} We investigate fast algorithms for several
fundamental connectivity problems in (weighted) undirected graphs as
well as their generalizations to the abstract setting of submodular
and bisubmodular functions. The motivation for this work arose from
the recent paper of \citet{li-panigrahi} that described a new
algorithmic approach for finding the \emph{global minimum cut} in an
undirected graph.  For a graph $G=(V,E)$ with edge weights
$w: E \to \preals$, the global minimum cut problem is to find the
minimum weight subset of edges whose removal disconnects the graph;
alternatively it is to find a set $S$, where
$\emptyset \subsetneq S \subsetneq V$, that minimizes
$w(\delta(S))$\footnote{For $A \subset V$, $\delta(A)$ denote the set
  of edges in $G$ with exactly one end point in $A$. $w(\delta(A))$ is
  notation for $\sum_{e \in \delta(A)} w(e)$.}.  When $G$ is
unweighted, this is called the edge connectivity of the graph.  There
has been extensive work on algorithms for this problem, and its study
has led to many important theoretical developments. Karger developed a
near-linear time randomized algorithm \cite{Karger00} that runs in
$O(m\log^3 n)$ time with some recent improvements in the log factors
via better data structures \cite{GMW-20,MN20}. Here $m$ is the number
of edges and $n$ is number of nodes in the graph.  However, the best
deterministic algorithm until recently was $\tilde{O}(mn)$ via two
very different approaches \cite{hao-orlin,StoerW97}.  Li and Panigrahi
developed a new approach that improved this bound. Their algorithm
runs in time $O(m^{1+o(1)})$ plus the time to compute $O(\polylog{n})$
$(s,t)$-minimum cut computations in a graph with $m$ edges and $n$
nodes. Their approach uses the $(s,t)$-minimum cut algorithm as a
black box.

\smallskip \noindent {\bf Isolating cuts:} A key technique in
\cite{li-panigrahi} is an algorithm to find \emph{isolating cuts}.  To
describe this notion, let $R \subseteq V$ be subset of nodes that we
call terminals. Given $r \in R$, a set $S \subseteq V$ is an isolating
cut for $r$ (with respect to $R$) if $S \cap R = \{r\}$. Consider the
problem of finding, for \emph{each} $r \in R$, a minimum weight
isolating cut, that is; a cut $S_r \subseteq V$ where
$S_r = \argmin_{S \subseteq V, S\cap R =\{v\}} w(\delta(S))$. Note
that if $R=V$ this is trivial since $S_r = \{r\}$ for each
$r$. However, the problem is non-trivial when $R \subset V$ is a
proper subset of $V$. A naive approach would require $|R|$
$(s,t)$-minimum cut computations. Li and Panigrahy described a simple
and elegant procedure that computes all the isolating cuts for any
given $R$ in time proportional to $O(\log |R|)$ $(s,t)$-minimum cut
computations.  This, combined with simple random sampling, can be used
to easily derive a randomized algorithm for global minimum cut that
relies on $O(\polylog{n})$ $(s,t)$-minimum cut computations. Note that even
though the total time corresponds to $O(\polylog{n})$ $(s,t)$-minimum cuts,
the second phase of their algorithm requires computing $|R|$
$(s,t)$-minimum cuts, but in smaller graphs whose total size is $O(m)$
and thus can be folded into a single $(s,t)$-minimum cut on roughly
the same input size as the original graph.  Their algorithm gives a
new randomized approach to global minimum cut; however, it does not
lead to a faster algorithm than the existing near-linear time
algorithm. Instead \cite{li-panigrahi} focuses on deterministic
running times and avoids random sampling by relying on several
technical tools including deterministic expander decompositions to
obtain a deterministic algorithm.  We note, however, that the
algorithm in \cite{li-panigrahi} applies to the more general problem
of finding the Steiner minimum cut: given $X \subseteq V$, the goal is
to find a minimum cut spearating a pair of nodes in $X$. See
\cite{HariharanTPB07,JueK19} for applications.

\smallskip
\noindent
{\bf Vertex and element connectivity:}
Our focus here is not on deterministic algorithms per se but rather on
the applicability of the isolating cut approach to derive faster
(randomized) algorithms in settings beyond edge connectivity.  There
has been tremendous recent and ongoing progress in fast algorithms for
$(s,t)$-flow and cut problems and leveraging these algorithms for
global connectivity is opened up by the new approach.  In particular,
an important motivating problem is to compute the global
\emph{(weighted) vertex connectivity} of a graph which has received
substantial recent attention \cite{fnsyy,nsy-19}. In this setting we
are given a graph $G=(V,E)$ with vertex weights
$w: V \rightarrow \mathbb{R}_+$ and the goal is to find a minimum
weight subset $S \subset V$ such that $G-S$ has at least two
non-trivial connected components.  However, as is well-known, vertex
cuts/separators are not as easy to work with as edge cuts. Despite
recent exciting progress via an approach based on local cuts and
connectivity, the weighted case had not been addressed and the best
known algorithms are from the work of Henzinger, Rao and Gabow
\cite{hgr}.  Our starting point is the observation that the isolating
cut approach of \cite{li-panigrahi} relies only on the submodularity
and symmetry of the edge-cut function of undirected graphs. Recall
that a real-valued set function $f:2^V \rightarrow \reals$ is
\emph{submodular} iff $f(A) + f(B) \ge f(A \cup B) + f(A \cap B)$ for
all $A,B \subseteq V$.  A set function is symmetric if
$f(A) = f(V\setminus A)$ for all $A \subseteq V$.  The applicability
of the isolation cut approach to symmetric submodular set functions
already yields faster algorithms for hypergraph connectivity and
several other problems that we describe subsequently. However, as we
already remarked, vertex cuts do not lend themselves to this approach
as vertex cuts, unlike undirected edge cuts, are simply not a
symmetric submodular function.

When considering isolating cuts in the context of vertex connectivity
one naturally encounters the notion of \emph{element connectivity},
which has been found to have several important connections between
edge and vertex connectivity. Element connectivity plays a key role in
network design, and in fact, it was introduced
by \citet{Jainetal02} to overcome the difficulty of working with
vertex connectivity. We refer the reader to surveys and related papers
on network design
\cite{FleischerJW06,CheriyanVV06,ChuzhoyK09,GuptaK11-survey,KortsarzN10-survey}
for extensive literature on this topic. It also plays an important
role in packing vertex disjoint Steiner trees and forests among others
\cite{CheriyanS07,CalinescuCV09,AazamiCJ12,ChekuriK14}; \cite{chekuri-survey} surveys this area. We now
formally define element connectivity. The input is a graph $G=(V,E)$
and a partition of $V$ into terminals $T$ and non-terminals $N =
V\setminus T$.  The \emph{elements} of $G$ are the edges and
non-terminals; that is, $E \cup N$.  For two terminals $s,t$ we define
the element connectivity between $s$ and $t$ as the minimum number of
elements whose removal disconnects $s$ from $t$. We emphasize that
element connectivity is defined \emph{only} between the terminals. We
can generalize this to the weighted setting where edges and
non-terminals have non-negative weights. The global element
connectivity of $G=(T \cup N,E)$ is the minimum element connectivity
between any two terminals. \citet{crx} considered algorithms for
computing (global) element connectivity.  For
global element connectivity they obtained an algorithm with running
time $O(|T|)$ times the time for $(s,t)$-minimum cut computation.

\medskip
\noindent
{\bf Set-pairs and Bisubmodularity:}
Cuts for element and vertex connectivity do not fall into
the setting of symmetric submodular set functions. A vertex separator
$S$ induces a partition of $V\setminus S$ into disjoint sets $A,B$
that do not share an edge, and obviously $B \neq V \setminus A$ (for
nonempty $S$). Nevertheless, one of the reasons for the tractability
of element connectivity is that it does admit submodularity
properties. The natural way to view its submodularity properties is
via the more general notion of \emph{bisubmodular} set
functions. Given a ground set $V$ a set-pair is $(A,B)$ where
$A, B \subseteq V$. Informally speaking a bisubmodular function $f$
assigns a real-value to each set-pair $(A,B)$ in a collection of
set-pairs as to satisfy the inequality
\begin{align*}
  \f{X_1,Y_1} + \f{X_2,Y_2} \geq %
  \f{X_1 \cup X_2, Y_1 \cap Y_2} + \f{X_1 \cap X_2, Y_1 \cup Y_2}
\end{align*}
for all set-pairs $(X_1,Y_1)$ and $(X_2,Y_2)$ on which it is
defined. For this to make sense the collection of set-pairs needs to
be closed under the above criss-crossed intersection and union
operations for set-pairs. These binary operations can be understood
more clearly as the meet and join of an appropriately defined lattice;
we defer the formal definitions to \refsection{bisubmod}.  One can
generalize the notion of cuts to set-pairs. Let $(S,T)$ be a set-pair
corresponding to a partition of a terminal set $R$.  A set-pair
$(A,B)$ cuts $(S,T)$ if $S \subseteq A$ and $T \subseteq B$. One can
then define the $f$-minimum cut problem for $(S,T)$: find the set-pair
of minimum $f$ value among all set-pairs that cut $(S,T)$. With this
definition in place the notions of global minimum cut for a terminal
set $R \subseteq V$, and isolating cuts for $R$, naturally generalize.
In this paper we show that the
isolating cut approach of \cite{li-panigrahi} generalizes to the class
of symmetric bisubmodular set functions defined over appropriate
collections of set-pairs.

\subsection{Contributions and Results}
We make two contributions at the high-level. The first is conceptual
in generalizing the isolating cut approach to the (bi)submodular
setting.  The second is to apply this abstract framework with
additional ideas to derive faster randomized algorithms for several
fundamental problems. Together they yield a plethora of new running
times for a diverse collection of connectivity problems, both abstract
(optimizing over set functions in an oracle model) and concretely in
graphs. The multiplicity of results is for the following combination
of reasons. First, by implementing the isolating cut approach at a
higher level of abstraction, and abstaining from concrete
specificities, we not only expose the isolating cut approach to new
problems, but allow for the substitution of different domain specific
black box subroutines that, within a domain, can have interesting
tradeoffs. Second, and unlike the case of graph edge connectivity, the
second phase of the isolating cut approach can often benefit from
additional problem specific ideas, especially if one wants to take
advantage of certain domain-specific algorithms that can be very
powerful if applied carefully.

An important aspect of the isolating cut approach is that it
inherently gives an algorithm for the subset connectivity version.  In
the following we will use $m,n$ to refer to the number of edges and
vertices in a given graph and use $\ectime{m,n}$ to refer to the
running time for computing a minimum $(s,t)$-cut in an edge-weighted
directed graph, and $\vctime{m,n}$ for the running time for computing
a minimum $(s,t)$-cut in a vertex-weighted directed graph.  We
instantiate concrete running times for special cases when needed.

\smallskip
\noindent
{\bf Connectivity of Bisubmodular functions:}
The precise statement that captures the general isolation cut property
in bisubmodular set functions requires stating several technical
definitions. Our main results for this are captured by
\reflemma{isolating-cut-partition}, \reflemma{isolating-cuts} and
\reftheorem{bisubmod-mincut} which are better understood after
the technical definitions. Here we state an informal theorem that
captures these results.

\begin{theorem}(Informal) Let $f: \V \rightarrow \reals$ be a
  symmetric bisubmodular function defined over a collection of
  set-pairs $\V$ over $V$. Let $R \subseteq V$. Suppose one has an
  oracle that given a partition $(S,T)$ of $R$ finds the $f$-minimum
  set-pair $(A,B) \in \V$ that cuts $(S,T)$. In $O(\log |R|)$ calls to
  this oracle one can find for each $r \in R$ a set-pair
  $(X_r,X'_r) \in \V$ such that the following properties hold: (i) for
  each $r \in R$, $(X_r,X'_r)$ is a $(r,R-r)$ separating set-pair,
  (ii) there is an $f$-minimum set-pair $(Y_r,Y'_r)$ separating
  $(r,R-r)$ such that $Y_r \subseteq X_r$, $X'_r \subseteq Y'_r$ and
  (iii) $X_r \cap X_q = \emptyset$ for $r \neq q$. The total run time
  for finding the $f$-minimum isolating cut $(Y_r,Y'_r)$ for each
  $r \in R$ can thus be bounded by the $O(\log |R|)$ cut computations
  and the total time to find the cuts inside each $(X_r,X'_r)$.
\end{theorem}

\smallskip
\noindent
{\bf Symmetric submodular functions:} We derive the following theorem
as a corollary.

\begin{theorem}
  \labeltheorem{intro-symsubmod-minimum cut}
  Let $\deff$ be a symmetric submodular function and $R \subseteq V$
  and let $n = |V|$. Suppose there is an algorithm for submodular
  function minimization in the value oracle model in time
  $\SFMtime{n} = g_1(n) \text{EO} + g_2(n)$ where $\text{EO}$ is the time for the
  evaluation oracle. Assuming that $g_1(n) = \Omega(n)$ and
  $g_2(n) = \Omega(n)$, a minimum $f$-cut that separates some two
  terminals in $R$ can be found in
  $\bigO{\SFMtime{n} \log^2 (n)}$ time.
\end{theorem}

\begin{corollary}
  Let $f$ be an integer valued symmetric submodular function with
  $|f(S)| \le M$. Using the submodular function minimization
  algorithms of \cite{lsw-15} one can find the global minimum cut of
  $f$ with high probability in time
  $\tilde{O}(n^2\log (nM) \text{EO} + n^3 \log^{O(1)} (nM))$.
\end{corollary}
The preceding corollary should be compared to Queyranne's well-known
combinatorial algorithm that uses $O(n^3 \text{EO})$ time \cite{q-98}.
The algorithm from \cite{lsw-15} is not strongly polynomial but uses a
factor $\tilde{\Omega}(n)$ fewer evaluation calls. Further, our
randomized algorithm can handle minimum $f$-cut for a subset of
terminals while Queyranne's algorithm does not generalize. In
addition, the black box reduction can take advantage of future
improvements to $\SFMtime{n}$ as well as for special cases as we will
see next.

\smallskip
\noindent {\bf Hypergraph connectivity:} A hypergraph $H=(V,E)$ consists of
vertices $V$ and hyperedges $E$ where each hyperedge $e \in E$ is a
subset of nodes; that is, $e \subseteq V$. We let
$p = \sum_{e \in E} |e|$ denote the total size of $H$ and let $m,n$
denote number of hyperedges and vertices. The rank $r$ of a hypergraph
is the maximum edge size; graphs are rank $2$ hypergraphs.
The cut function of a hypergraph is symmetric and submodular and the
global minimum cut question for edge connectivity naturally
generalizes to hypergraphs.  The best deterministic algorithm for this
problem runs in $O(pn+n^2 \log n)$ time \cite{KlimmekW96,q-98,MakW00}.
The best randomized algorithm runs in time $\tilde{O}(n^r)$ time with
high probability in rank $r$ hypergraphs \cite{FoxPZ19} and this is
better than $\tilde{O}(pn)$ only for very dense hypergraphs.  Via
sparsification one can also get an algorithm in unweighted hypergraphs
that runs in time $O(p + \lambda n^2)$ where $\lambda$ is the minimum
cut value \cite{cx-18}.  We obtain the following theorem that gives
significantly better bounds in most settings of interest, and new
tradeoffs, while also generalizing to subset minimum cut.

\begin{theorem}
  \labeltheorem{intro-hypergraphs} Let $H=(V,E)$ be a weighted
  hypergraph with $m$ edges, $n$ nodes and total size
  $p = \sum_{e \in E} |e|$. Let $R \subseteq V$.  The global minimum
  cut for $R$ in $H$ can be found with high probability in time
  $\tilde{O}(\ectime(p,m+n))$ or in time
  $\tilde{O}(\sqrt{pn(m+n)^{1.5}})$.
\end{theorem}

Now we state our algorithmic results for element connectivity and vertex
connectivity that follow via the bisubmodularity framework and
problem specific optimizations.

\smallskip
\noindent {\bf Element connectivity:} The fastest known algorithm so
far for global element connectivity is from \cite{crx} and runs in
time $O(|T| \ectime{m,n})$ for terminal set $T$, which can be
$\Omega(n \ectime{m,n})$. We obtain the following.
\begin{theorem}
  \label{thm:intro-elem-conn}
  Let $G=(T \cup N, E)$ be an instance of weighted element
  connectivity with $|T| =  k$ terminals. The global element
  connectivity can be computed in
  \begin{math}
    \apxO{
    \ectime{m}{n} + \max[{m_1,\dots,m_k}]{\sum_{i=1}^k
      \ectime{m_i}{n} \where m_1 + \cdots + m_k \leq 2 m}
    }
  \end{math}
  time with high probability. The algorithm generalizes to subset
  element connectivity.
\end{theorem}
In particular, for $\ectime{m}{n}$ of the form $\ectime{m}{n} =\apxO{m
  \poly{m,n}}$, the running time above is $\apxO{\ectime{m}{n}}$. For
instance, via \cite{lee-sidford}, one obtains an $\apxO{m \sqrt{n}}$ time
algorithm. However, recent breakthrough work of \cite{brand+} showed
that $\ectime{m}{n} = \apxO(m + n^{1.5})$. This running time bound
cannot be directly used in the preceding theorem. Using further ideas
we obtain an improved running times that are encapsualted in the
following theorem.
\begin{theorem}
  Let $G = (T \cup N, E)$ be an instance of
  element connectivity with $n$ nodes and $m$ edges.
  Let $w: V \cup E \to [1..U]$ assign integer
  (or infinite) weights to each vertex and edge. The global element
  connectivity can be computed in randomized  $\apxO{m^{1 + o(1)} n^{3/8} U^{1/4} + n^{1.5}}$
  time or in $ \apxO{m^{1/2} n^{5/4}}$ time where  $\apxO{\cdots}$ hides
  $\poly{\log{n}, \log{U}}$-factors.
\end{theorem}

\smallskip
\noindent {\bf Vertex connectivity:} We now consider global vertex
connectivity of both weighted and unweighted graphs. For simplicity we
consider the interesting setting where there is a vertex separator of
size less than $0.99W$ where $W$ is the total vertex weight.
We obtain new and faster randomized $(1+\eps)$-approximation algorithms
that improves upon the randomized
$\tilde{O}(mn)$ exact algorithm of Henzinger, Rao and Gabow
\cite{hgr}. The algorithms are based on reducing, via sampling,
to computing isolating element cuts. The running times we
obtain are captured by the following theorem.

\begin{theorem}
  \label{thm:intro-vertex-conn}
  Let $G=(V,E)$ be a \emph{weighted} instance  of vertex connectivity.
  There is a randomized algorithm that gives a
  $(1+\eps)$-approximation with high-probability in
  $\tilde{O}(\ectime{m}{n}/\eps)$ time; in particular there is a randomized algorithm
  that runs in time $\tilde{O}(m \sqrt{n}/\eps)$. For dense graphs
  there is a randomized algorithm that runs in  $\tilde{O}(m^{1/2}n^{5/4}/\eps)$ time.
\end{theorem}

There has been exciting recent work on faster algorithms for vertex
connectivity via a local connectivity approach
\cite{nsy-19,fnsyy}. The algorithms are limited to unweighted graphs
while our theorem above gives the first constant factor approximation
for weighted vertex connectivity in $o(mn)$ time. For unweighted
graphs and graphs with small integer capacities we can obtain
\emph{exact} algorithms by setting $\eps = 1/\kappa$ where $\kappa$ is
the vertex connectivity. We obtain several different tradeoffs
depending on $m,n, \kappa$. These can be found in \refsection{vc}.

\medskip
\noindent {\bf Organization:}
\refsection{bisubmod-short} describes the bisubmodularity framework and the
abstract results; we defer to the appendix (\refsection{bisubmod}) a more detailed description with
several examples and formal proofs of the lemmas and theorems stated in \refsection{bisubmod-short}.
\refsection{elem-conn} describes the algorithms for
element-connectivity and \refsection{vc} describes our algorithms for
vertex connectivity. \refsection{hypergraphs} describes our algorithms for hypergraph
connectivity.


\section{Isolating Cuts, Symmetric Bisubmodular Functions, and
  Lattices}
\labelsection{bisubmod-short}

Our goal in this section to define the relevant machinery to explore
and make explicit the generality of the isolating cut idea. As
discussed in the introduction, this framework is motivated by the
necessity of going beyond symmetric submodular set functions to
capture concrete applications of interest such as element and vertex
connectivity. Given the abstract nature of this discussion, in
constrast to the concrete algorithmic applications, we have elected to
give a brief and minimal discussion of the bisubmodular framework
here, and have placed a more comprehensive description in the
appendix, in \refsection{bisubmod}. The appendix includes many more
examples as well as the proofs of all lemmas and theorems stated here.

Let $V$ be a finite set of elements. An ordered pair
$(A,B) \in 2^V \times 2^V$ is a \defterm{set-pair} over $V$.  For a
family of set-pairs $\defV$ over $V$, we say that $\V$ is a
\defterm{crossing lattice}\footnote{This notion is analogous to the
  definition of a crossing family of sets.}  over $V$ if it is closed
under the following two operators.
\begin{gather*}
  (X_1,Y_1) \lor (X_2,Y_2) = (X_1 \cup X_2, Y_1 \cap Y_2). \\
  (X_1,Y_1) \land (X_2,Y_2) = (X_1 \cap X_2, Y_1 \cup Y_2).
\end{gather*}
If $\V$ is closed under these operations, then $\V$ is a lattice
under the partial order
\begin{align*}
  (X_1,Y_1) \preceq (X_2,Y_2) \iff X_1 \subseteq X_2 \andcomma Y_2
  \subseteq Y_1.
\end{align*}
The binary operator $\lor$ returns the unique least upper bound of
its arguments (a.k.a.\ the \defterm{meet}) and the binary operator
$\land$ returns the unique greatest lower bound of its arguments
(a.k.a.\ the \defterm{join}).

For a pair of sets $(X,Y) \in \subsetsof{V} \times \subsetsof{V}$, the
\defterm{transpose} of $(X,Y)$, denoted $(X,Y)^T$, is the reversed
pair of sets $(X,Y)^T \defeq (Y,X)$. A crossing lattice $\defV$ is
\defterm{symmetric} if is closed under taking the transpose.  We have
the following identities relating the transpose with the lattice
operations $\lor$ and $\land$. Observe that for
$\X, \Y \in
\V$, we have
\begin{math}
  \parof{\X^T}^T
  = \X,
\end{math}
\begin{math}
  \parof{\X \lor \Y}^T = \X^T \land \Y^T,
\end{math}
and
\begin{math}
  \parof{\X \land \Y}^T = \X^T \lor \Y^T.
\end{math}
Lastly, A crossing lattice $\defV$ is \defterm{pairwise disjoint} if
$X \cap Y = \emptyset$ for all $(X,Y) \in \V$.

We now define an abstract, lattice-based notion of cuts that unifies
the various different families of cuts of interest in graphs. Let $V$
be a set.  For two set-pairs
$\S = (S,T) \in \subsetsof{V} \times \subsetsof{V}$ and
$\X = (X,Y) \in \subsetsof{V} \times \subsetsof{V}$, we denote
\begin{align*}
  \S \subseteq \X \defiff S \subseteq X \andcomma T \subseteq Y.
\end{align*}
If $\S \subseteq \X$, then we say that $\X$ \defterm{cuts} $\S$ or
that $\X$ is an \defterm{$\S$-cut}. If $\V$ is a crossing lattice over
$V$, $R \subset V$ is a subset, and $\R$ is a crossing lattice over
$R$, then we say that $\V$ \defterm{separates} $\R$ if for every
$\S \in \R$, there is an $\S$-cut $\X \in \V$.  The following lemma
observes that cuts are closed under the two lattice operations. 

\begin{restatable}{lemma}{LatticeCuts}
  \labellemma{lattice-cuts} Let $V$ be a set and let $R \subseteq
  V$. Let $\defV$ be a crossing lattice over $V$ and let $\defR$ be a
  crossing lattice over $R$. Suppose that $\V$ separates $\R$. Let
  $\S_1,\S_2 \in \R$, let $\X_1 \in \V$ be an $\S_1$-cut, and let
  $\X_2 \in \V$ be an $\S_2$-cut.  Then $\X_1 \lor \X_2$ is an
  $\S_1 \lor \S_2$-cut and $\X_1 \land \X_2$ is an
  $\S_1 \land \S_2$-cut.
\end{restatable}

Now, let $\V$ be a lattice. A real-valued function $f: \V \to \reals$
is \defterm{submodular} if for all $\X,\Y \in \V$,
\begin{align*}
  \f{\X} + \f{\Y} \geq \f{\X \lor \Y} + \f{\X \land \Y}.
\end{align*}
\emph{Bisubmodular functions} can be interpreted as submodular
  functions over particular crossing lattices. There are at least two
  definitions of bisubmodular function in the literature. These
  definitions are similar and we discuss both.

  In one definition (e.g., in \cite{schrijver}), a function
  $f: \subsetsof{V} \times \subsetsof{V} \to \reals$ is called
  \emph{bisubmodular} if for all $X_1,Y_1,X_2,Y_2 \subseteq V$, we
  have
  \begin{align*}
    \f{X_1,Y_1} + \f{X_2,Y_2} \geq %
    \f{X_1 \cup X_2, Y_1 \cap Y_2} + \f{X_1 \cap X_2, Y_1 \cup Y_2}.
    \labelthisequation{bisubmodular}
  \end{align*}  A bisubmodular function
  $f: \subsetsof{V} \times \subsetsof{V} \to \reals$ is submodular
  over the crossing lattice of all set-pairs,
  $\V = \subsetsof{V} \times \subsetsof{V}$ (\refexample{all-sets}).

  Another definition (e.g.,
  \cite{bouchet-87,afn-96,af-96,fujishige-iwata}) of a bisubmodular
  function $f$ is that $f(X_1,Y_1)$ is only defined for disjoint sets
  $X_1$ and $Y_1$, and otherwise satisfies inequality
  \refequation{bisubmodular} for these inputs. In this version, $f$ is
  bisubmodular iff it is a submodular function over the lattice of
  disjoint sets,
  $\V = \setof{(X,Y) \where X,Y \subseteq V \andcomma X \cap Y =
    \emptyset}$ (\refexample{disjoint-sets}).

Now, let $\defV$ be a symmetric crossing lattice. A function
$f: \V \to \reals$ is \defterm{symmetric} if for all $\X \in \V$,
\begin{math}
  \f{\X} = \f{\X^{T}}.
\end{math}
This is a different definition then for symmetric submodular
\emph{set} functions and generalizes the (more standard) set-based
definition. Both undirected edge cuts and vertex cuts are examples of
symmetric submodular functions over appropriate symmetric crossing
lattices.

There is an important relationship between the sets of terminals being
separated and \emph{minimal} minimum cuts that separate them,
highlighted in the following lemma. See \reffigure{lattice-min-cuts}
for a diagrammatic description of the following lemma.

\begin{restatable}{lemma}{LatticeMinCuts}
  \labellemma{lattice-min-cuts} Let $V$ be a set and $R \subset
  V$. Let $\defV$ be a symmetric crossing lattice over $V$ and let
  $\defR$ be a symmetric crossing lattice over $R$, such that $\V$
  separates $\R$.  Let $f: \V \to \reals$ be a symmetric bisubmodular
  function. Consider the function $h: \R \to \V$ where $h(\S)$ is
  defined as $\preceq$-minimum, $f$-minimum $\S$-cut. Then $h$ is
  well-defined and carries the partial orders on $\R$ to $\V$; that
  is, $\S_1 \preceq \S_2 \implies h(\S_1) \preceq
  h(\S_2)$.
\end{restatable}

The following is a particularly convenient form of
\reflemma{lattice-min-cuts}, and the one applied directly in the
sequel. A diagram depicting the following lemma is given in
\reffigure{uncrossing-min-cuts}. The proof can be found in the
appendix.

\begin{restatable}{lemma}{UncrossingMinCuts}
  \labellemma{uncrossing-min-cuts} Let $V$ be a set and $R \subset
  V$. Let $\defV$ be a symmetric crossing lattice over $V$ and let
  $\defR$ be a symmetric crossing lattice over $R$, such that $\V$
  separates $\R$.  Let $f: \V \to \reals$ be a symmetric bisubmodular
  function. Let $\S_1,\dots,\S_k \in \R$ and $\X_1,\dots,\X_k \in \V$
  such that for all $i \in [k]$, $\X_i$ is an $f$-minimum
  $\S_i$-cut. Then for any $\S \in \R$ such that $\S \preceq \S_i$ for all $i$,
  there is an $f$-minimum $\X$-cut with
  $\X \preceq \X_1 \land \cdots \land \X_k$.
\end{restatable}

We now come to the issue of computing isolating cuts. We formalize
this as follows. Let $V$ be a set and $R \subset V$. Let $\defV$ be a
symmetric and \emph{pairwise disjoint} crossing lattice over $V$ and
let $\defR$ be the symmetric and \emph{pairwise disjoint} crossing
lattice over $R$ consisting of all partitions of $R$; i.e.,
\begin{math}
  \R = \setof{(S,T) \where S \cup T = R \andcomma S \cap T =
    \emptyset}.
\end{math}
Let $f: \V \to \reals$ be a symmetric bisubmodular function.  For
each $r \in R$ we wish to find an $f$-minimum cut $\Y_r$ for the set-pair
$(\{r\}, R-\{r\})$ (which we abbreviate as $(r,R-r)$ for notational
simplicity). The main property that leads to efficiency
is captured by the next lemma.

\begin{restatable}{lemma}{IsolatingCutPartition}
  \labellemma{isolating-cut-partition} Let $V$ be a set and
  $R \subset V$. Let $\defV$ be a symmetric and \emph{pairwise disjoint}
  crossing lattice over $V$ and let $\defR$ be the symmetric and
  \emph{pairwise disjoint} crossing lattice over $R$ consisting of all
  partitions of $R$; i.e.,
  \begin{math}
    \R = \setof{(S,T) \where S \cup T = R \andcomma S \cap T =
      \emptyset}.
  \end{math}
  Let $f: \V \to \reals$ be a symmetric bisubmodular function.
  Suppose we had access to an oracle that, given $\S \in \R$, returns
  a minimum $\S$-cut $\W \in \V$. Let $k = \logup{\sizeof{R}}$. Then
  with $k$ calls to the oracle, one can compute $k$ cuts
  $\W_1,\dots,\W_k \in \V$ such that the following holds.

  For each $r \in R$, let
  \begin{math}
    \X_r = \parof{\Land_{i \where (r, V-r) \preceq \W_i} \W_i} %
    \land %
    \parof{\Land_{i \where (r, V-r) \preceq \W_i^T} \W_i^T}
  \end{math}
  be the intersection of cuts transposed to always include $r$ in the
  first component. Then we have the following.  (1) For all $r \in R$,
  $\X_r$ is an $(r, R-r)$-cut. (2) For all $r \in R$, there is a
  minimum $(r, R - r)$-cut $\Y_r$ such that $\Y_r \preceq \X_r$. (3)
  For any two distinct elements $r,q \in R$,
  $\X_r \land \X_q \preceq \parof{\emptyset, R}$. (That is, the first
  components of the set pairs $\X_r$ are pairwise disjoint.)
\end{restatable}

Using the preceding lemma the problem of computing the $f$-minimum
$r$-isolating cuts is reduced to finding such a cut in $\X_r$. The
advantage, in terms of running time, is captured by the disjointness
property: for distinct $r, q \in R$ we have
$\X_r \land \X_q \preceq (\emptyset, R)$.  For each $r$ let
$\X_r = (A_r, B_r)$. Thus we have $\sum_r |A_r| \le |V|$.  Given $r$
and $\X_r$, the problem of computing the $f$-minimum cut
$\Y_r \preceq \X_r$ can in several settings be reduced to solving a
problem that depends only on $|A_r|$ and $|V|$. We capture this in the
following lemma.

\begin{restatable}{lemma}{IsolatingCuts}
  \labellemma{isolating-cuts} Let $V$ be a set and
  $R \subset V$. Let $\defV$ be a symmetric and \emph{pairwise disjoint}
  crossing lattice over $V$ and let $\defR$ be the symmetric and
  \emph{pairwise disjoint} crossing lattice over $R$ consisting of all
  partitions of $R$; i.e.,
  \begin{math}
    \R = \setof{(S,T) \where S \cup T = R \andcomma S \cap T =
      \emptyset}.
  \end{math}
  Let $f: \V \to \reals$ be a symmetric bisubmodular function.
  Suppose we had access to an oracle that, given $\S \in \R$, returns
  a minimum $\S$-cut $\W \in \V$ and let $\SMtime{n}$ denote its
  running time where $n=|V|$.  Moreover, suppose we have an oracle
  that given any $u \in R$ and $(A_u,B_u) \in \V$ with $u \in A_u$
  outputs an $f$-minimum cut $\Y_u \preceq (A_u,B_u)$ in time
  $\SMItime{|A_u|,n}$.  Let $k = \logup{\sizeof{R}}$.  Then, one can
  compute for each $r \in R$ an $f$-minimium $r$-isolating cut in
  in total time
  \begin{math}
    O(k \SMtime(n) + \max_{0 \le n_1,n_2,\ldots,n_{|R|}: \sum_i n_i =
      n} \sum_{i=1}^{|R|}\SMItime{n_i,n}).
  \end{math}
\end{restatable}

A simple random sampling approach combined with isolating cuts, as
shown in \cite{li-panigrahi} for edge cuts in graphs, yields the
following theorem in a much more abstract setting.

\begin{restatable}{theorem}{BisubmodMincut}
  \labeltheorem{bisubmod-mincut}
  Let $V$ be a set and
  $R \subset V$. Let $\defV$ be a symmetric and pairwise disjoint
  crossing lattice over $V$ and let $\defR$ be the symmetric and
  pairwise disjoint crossing lattice over $R$ consisting of all
  disjoint subsets of $R$; i.e.,
  \begin{math}
    \R = \setof{(S,T) \where S, T \subseteq R \andcomma S \cap T = \emptyset}.
  \end{math}
  Let $f: \V \to \reals$ be a symmetric bisubmodular function.
  Suppose we had access to an oracle that, given $\S \in \R$, returns
  a minimum $\S$-cut $\W \in \V$ and let $\SMtime{n}$ denote its
  running time where $n=|V|$.  Moreover, suppose we have an oracle
  that given any $u \in R$ and $(A_u,B_u) \in \V$ with $u \in A_u$
  outputs an $f$-minimum cut $\Y_u \preceq (A_u,B_u)$ in time
  $\SMItime{|A_u|,n}$.  Then one can
  compute the minimum (nontrivial) $\R$-cut with constant probability
  in
  \begin{math}
    \bigO{\SMtime{n} \log^2{\sizeof{R}} + \max_{0 \le
    n_1,n_2,\ldots,n_{|R|}:  \sum_i n_i = n} \log{\sizeof{R}}\sum_{i=1}^{|R|}\SMItime{n_i,n}}
  \end{math}
  time.
\end{restatable}

We derive the following corollary for symmetric submodular set
functions. 

\begin{restatable}{corollary}{SymsubmodMincut}
  \labelcorollary{symsubmod-mincut}
  Let $\deff$ be a symmetric submodular function and $R \subseteq V$
  and let $n = |V|$. Suppose there is an algorithm for submodular
  function minimization in the value oracle model in time
  $\SFMtime{n} = g_1(n) \text{EO} + g_2(n)$ where $\text{EO}$ is the time for the
  evaluation oracle. Assuming that $g_1(n) = \Omega(n)$ and
  $g_2(n) = \Omega(n)$, a minimum $f$-cut that separates some two
  terminals in $R$ can be found in
  $\bigO{\SFMtime{n} \log^2 (n)}$ time.
\end{restatable}


\section{Element connectivity}
\labelsection{elem-conn}

\setupG Let $T \subseteq V$ be a set of terminals and let
$N= V\setminus T$ be the non-terminal set.  For any two distinct
terminals $u,v \in T$, the element connectivity between $u$ and $v$ is
defined as the maximum number of paths from $u$ to $v$ that are
edge-disjoint and vertex-disjoint in the non-terminal vertices
$V \setminus T$. That is, only terminal vertices may be reused across
paths. This notion can be easily generalized to the weighted setting
where edges and non-terminals have non-negative weights/capacities.
For any two terminals $s, t \in T$, we denote by $\elemconn(s,t)$ the
element connectivity between them.  One can compute $\elemconn(s,t)$
via a simple reduction to $s$-$t$ maximum flow in a directed graph
which takes $\ectime{m,n}$ time.
In this section we
are concerned with the problem of computing the global element
connectivity which is defined as
$\elemconn = \min_{s,t \in T, s \neq t} \elemconn(s,t)$. In fact we
are also interested in computing the more general problem of computing
$\elemconn(R) = \min_{s,t \in R, s \neq t} \elemconn(s,t)$ where
$R \subseteq T$; note that $\elemconn = \elemconn(T)$.
Here we apply our general
framework that obtains a randomized algorithm with running time
$O(\ectime{m,n} \log^2 |R|)$. In addition to the global minimum cut for
$R$, as we will see in the next section, finding all the isolating
cuts can be used with other ideas for vertex connectivity.

Let $G=(T \uplus N, E)$ be an instance of a weighted element
connectivity problem. Let $w: N \cup E \rightarrow \nnreals$ assign
weights to the elements.  Let $R \subseteq T$ be a subset of
terminals with $|R| \ge 2$. We reduce the  problem of computing
$\elemconn(R)$ to \reftheorem{bisubmod-mincut} as follows.

For ease of notation, let
$\barV = \elements$ denote the elements. Consider the family of pairs
of sets $\V \subseteq 2^{\barV} \times 2^{\barV}$ defined as the set
of pairs $(X,Y) \in \barV \times \barV$ with the following properties:
(i) $X$ and $Y$ are disjoint, (ii) no edge in $X$ is adjacent to a vertex in $Y$, and no edge in
$Y$ is adjacent to a vertex in $X$, and (iii) $T \subseteq X \cup Y$.
$\V$ describes the disjoint sets that are element-wise disconnected
and cover $T$. Clearly $\barV$ is symmetric and pairwise
disjoint. It is also straightforward to verify that $\barV$ is an
uncrossing lattice.

We define a function $\f: \V \to \reals$ by $\f{X,Y} = \sum_{x \in \barV - (X \cup Y)} \weight{x}$
$\f{\X}$ gives the total weight of elements that are not a member of
either of the two sets in $\X$. This function $\f$ is submodular
and in fact it is modular. One can easily verify that
$\f{X_1,Y_1} + \f{X_2,Y_2} =  \f{X_1 \cup X_2, Y_1 \cap Y_2} + \f{X_1 \cap X_2, Y_1 \cup Y_2}$.
Thus $\V$ is a symmetric and pairwise disjoint crossing lattice, and
$\f: \V \to \reals$ is a symmetric submodular function over $\V$.

\paragraph*{Isolating (weighted) element cuts and global
  connectivity}
Let $R \subseteq T$ and let $\R$ be the crossing lattice consisting
of all pairwise disjoint subsets of $R$.  Given a partition of $R$
into two sets $(A,B)$, an $\f$-minimum $(A,B)$-cut, which corresponds
to the minimum element cut separating $A$ from $B$, can
be computed via directed $(s,t)$-maxflow, in $\ectime{m,n}$-time.
By \reflemma{isolating-cut-partition}, we can
compute disjoint sets of elements
\begin{math}
  \setof{\barU_r \subset \barV \where r \in R}
\end{math}
where for each $r \in R$, $\barU_r$ contains the $r$-side component of a
miminimum $(r,R-r)$-element cut. Moreover, because the $\barU_r$'s
are obtained as intersections of sides of element cuts, for any
distinct $r,q \in R$, there is no edge from $\barU_r$ incident to a
vertex from $\barU_q$ and (symmetrically) vice-versa.

For each $r$, let $\barU_r' \subset \barV$ be the set of vertices
outside $\barU_r$ and incident to an edge in $\barU_r$, and the edges
outside $\barU_r$ incident to vertices in $\barU_r$. Informally
speaking, $\barU_r'$ is the ``boundary'' of $\barU_r$ in an element
connectivity sense. Let
$n_r = \sizeof{V \cap \parof{\barU_r \cup \barU_r'}}$ be the number of
vertices in $\barU_r \cup \barU_r'$ and let
$m_r = \sizeof{E \cap \parof{\barU_r \cup \barU_r'}}$ be the number of
edges in $\barU_r \cup \barU_r'$.  Note that $\sum_r m_r \leq 2 m$
since each edge can either appear in $\barU_r$ for a unique choice of
$r$ or in $\barU_r'$ for two choices of $r$.

To find an isolating cut for $r$ we need to find the cheapest element
cut contained in $\barU_r$. We can do this via a flow computation as
described below.  For each $r$, consider the graph $G_r$ where we
first take the graph $\barU_r \cup \barU_r'$ and introduce an
auxiliary vertex $\bar{t}$. We connect $\bar{t}$ to all vertices in
$\barU_r'$ with infinite capacity.
For every edge $e \in \barU_r'$ with exactly one endpoint in
$\barU_r$, we replace the opposite endpoint with $\bar{t}$. Observe
that the minimum $(r,\bar{t})$-element cut in $G_r$ coincides with the
minimum $(r,R-r)$-element cut in $G$. $G_r$ has $\bigO{m_r}$ edges and
$\bigO{n_r}$ vertices, and the element $(r,\bar{t})$-cut problem can
be solved in $\ectime{m_r}{n_r}$ time. Summing over all $r \in R$
gives the following theorem.

\begin{theorem}
  \labeltheorem{element-isolating-cuts} %
  Let $G=(T \cup N, E)$ be an instance of element connectivity with
  $n$ nodes and $m$ edges and let $R \subseteq T$. Let
  $\weight: \elements \to \pextreals$ assign positive weight to each
  vertex and edge. Let $k = \sizeof{T}$. Then one can compute, for all
  $r \in R$, the minimum weight element $(r,R-r)$-cut in
  \begin{math}
    \bigO{\ectime{m}{n} \log k + \max[{m_1,\dots,m_k}]{\sum_{i=1}^k
        \ectime{m_i}{n} \where m_1 + \cdots + m_k \leq 2 m}},
  \end{math}
  where $\ectime{m}{n}$ is the running time for element $(S,T)$-cut
  with $m$ edges and $n$ vertices.
\end{theorem}

With \reftheorem{element-isolating-cuts} in place we can reduce
the global mincut problem for $R$ to the isolating cut computation
via sampling \cite{li-panigrahi}, and captured in the abstract setting
\reftheorem{bisubmod-mincut}, to obtain the following theorem to
compute $\elemconn(R)$.

\begin{theorem}
  \labeltheorem{global-elem-mincut}
  Let $G=(T \cup N, E)$ be an instance of element connectivity with
  $n$ nodes and $m$ edges and let $R \subseteq T$. Let
  $\weight: \elements \to \pextreals$ assign positive weight to each
  vertex and edge. Let $k = \sizeof{R}$. Then one can compute $\elemconn(R)$
  with constant probability in time
  \begin{math}
    \bigO{                      %
      \parof{               %
        \ectime{m}{n} \log{k} %
        +                                 %
        \max[{m_1,\dots,m_n}]{\sum_{i=1}^n \ectime{m_i}{n} \where m_1 + \cdots +
          m_n \leq 2m}                 %
      }\log n          %
    }.
  \end{math}
\end{theorem}

\subsection{Refined running times for element connectivity}

\begin{figure}
  \centering
  \includegraphics{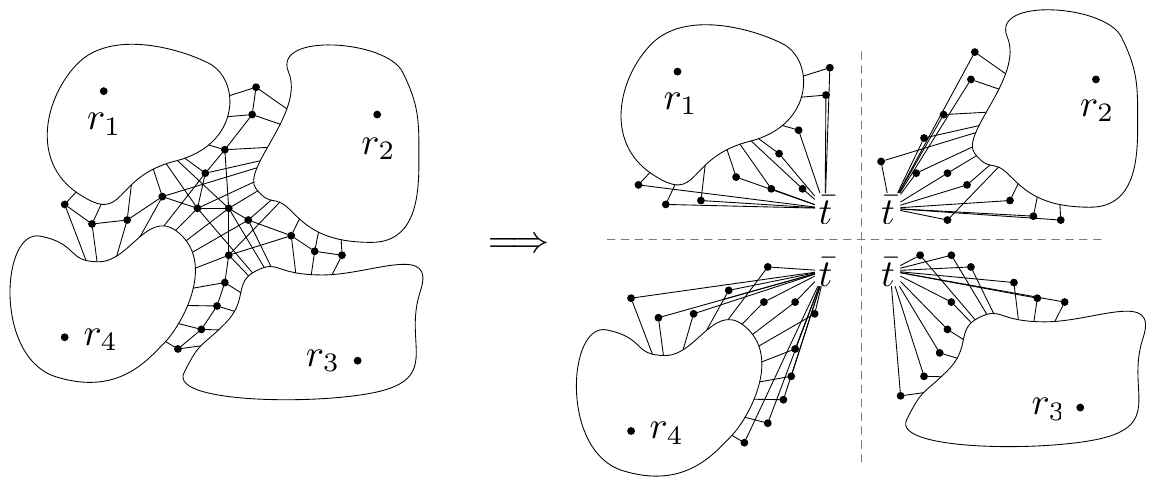}
  \caption{Applying the uncrossing framework to vertex connectivity
    reduces vertex isolating cuts to edge disjoint cut problems. Note
    that the separating vertices may appear in multiple subproblems,
    which is an obstruction towards a direct $\apxO{\ectime{m,n}}$
    overall running time for isolating vertex
    cuts.\labelfigure{isolating-vertex-cuts}}
\end{figure}

Until recently, the leading running times for $\ectime{m}{n}$ (e.g.,
$\ectime{m}{n} = \apxO{m \sqrt{n}}$ \cite{lee-sidford}) plug directly
into \reftheorem{global-elem-mincut} to give running times of the form
$\apxO{\ectime{m}{n}}$ to compute the global element connectivity.  A
recent breakthrough work by \citet{brand++} has obtained a running
time of
\begin{math}
  \ectime{m}{n} = \apxO{m + n^{1.5}}
\end{math}
for polynomially bounded and integral capacities.  However, \reftheorem{global-elem-mincut}
does not directly benefit from this running time because the vertices
are not partitioned across subproblems. See
\reffigure{isolating-vertex-cuts} for an illustration in the concrete
setting of vertex cuts.  Consequently, plugging
$\ectime{m}{n} = \apxO{m + n^{1.5}}$ directly into
\reftheorem{global-elem-mincut} generates a running time of
$\apxO{m + n^{1.5} k}$, where $k = \sizeof{T}$. The additional factor
of $k$ (to a certain extent) defeats the purpose of the isolating cuts
framework.

In this section, we develop more advanced algorithms that take the
isolating cut framework as a starting point, and incorporates
additional ideas to take advantage of
$\ectime{m}{n} = \apxO{m + n^{1.5}}$. In addition to obtaining faster
algorithms, these results point to a general algorithm design space
where additional ideas can be introduced to obtain even better running
times.  The first algorithm we present leverages the fact that the
edges are partitioned across subproblems, even if the vertices are
not.

\begin{restatable}{theorem}{IsoECTwo}
  Let $G = (T \cup N, E)$ be an instance of
  element connectivity with $n$ nodes and $m$ edges and let
  $R \subseteq T$. Let $w: V \cup E \to [1..U]$ assign integer
  (or infinite) weights to each vertex and edge. For $R \subseteq T$,
  the minimum $R$-isolating vertex cut can be
  computed in
  \begin{align*}
    \apxO{m^{1 + o(1)} n^{3/8} U^{1/4} + n^{1.5}}
  \end{align*}
  time.
\end{restatable}

\begin{proof}
  Let $k = \sizeof{R}$. We apply
  \reftheorem{element-isolating-cuts} and give concrete upper
  bounds using known upper bounds for $\ectime{m}{n}$.  Let
  $m_1,\dots,m_k \in \naturalnumbers$ with
  $m_1 + \cdots + m_k \leq m$. Recall that
  \begin{math}
    \ectime{m}{n} = \apxO{m^{4/3 + o(1)} U^{1/3}}
  \end{math}
  by \cite{liu-sidford-20-b} and
  \begin{math}
    \ectime{m}{n} = \apxO{m + n^{3/2}}
  \end{math}
  by \cite{brand++}. Let $\alpha > 0$ be a parameter to be
  determined. We apply the first running time when
  $m_i < m \alpha / k$ and the second running time then
  $m_i \geq m \alpha / k$. At most $k / \alpha$
  indices $i$ have $m_i \geq m \alpha / k$. Thus,
  \begin{align*}
    \sum_{i=1}^k \ectime{m_i}{n}    %
    &=                           %
      \sum_{i \where m_i \geq m \alpha / k} \ectime{m_i}{n}
      +                           %
      \sum_{i \where m_i < m \alpha / k} \ectime{m_i}{n}
    \\
    &\leq                       %
      \apxO{m + \frac{k}{\alpha} n^{1.5} + \sum_{i \where m_i < m \alpha /
      k} m_i^{4/3 + o(1)} U^{1/3}}
    \\
    &\tago{\leq}
      \apxO{m + \frac{k}{\alpha} \parof{n^{1.5} + \prac{\alpha m}{k}^{4/3+
      o(1)} U^{1/3}}}.
  \end{align*}
  Here \tagr is by convexity: the quantity
  $\sum_{i \where m_i < m \alpha / k} m_i^{4/3 + o(1)}$ subject to the
  condition that $\sum_i m_i \le m$ is at most
  $(k/\alpha) (\alpha m/k)^{4/3 + o(1)}$.  Balancing terms at
  $\alpha = k n^{9/8} / m$, this gives an upper bound of
  \begin{math}
    \apxO{m^{1 + o(1)} n^{3/8} U^{1/4}},
  \end{math}
  hence the claimed running time.
\end{proof}

We point out that other running time tradeoffs between $m$ and $U$ can
be obtained by instead applying the flow alogrithms from
\cite{liu-sidford-20-a, madry-16}.

The next theorem, which is particularly good for dense graphs,
leverages the fact that while the vertices are not necessarily
partitioned across subproblems, at least the ``inner'' vertex sets
$\barU_r \cap V$ are disjoint and all of the repeating ``boundary'' vertices
are guaranteed to be outside the $r$-component of each
$(r,R-r)$-minimum cut. The following algorithm balances a tradeoff
between the recent algorithm with \cite{brand+} with blocking flows
\cite{goldberg-tarjan-90}.  In the application of blocking flows, we
argue that with an appropriate construction of the auxiliary graph in
the components given by the decomposition by isolating cuts, the
maximum length of any augmenting paths is proportional to the number
of inner vertices (rather than the total number of vertices) for that
component.

\begin{restatable}{theorem}{IsoECThree}
  \labeltheorem{iso-ec-3} Let $G = (T \cup N, E)$ be an instance of
  element connectivity with $n$ nodes and $m$ edges and let
  $R \subseteq T$. Let $w: V \cup E \to [0,U]$ assign positive
  (or infinite) weights to each vertex and edge. For $R \subseteq T$,
  the minimum $R$-isolating cut can be computed in
  \begin{align*}
    \apxO{m^{1/2} n^{5/4}}
  \end{align*}
  randomized time, where $\apxO{\cdots}$ hides
  $\poly{\log{n}, \log{U}}$-factors.
\end{restatable}

\begin{proof}
  We recall the construction from \reftheorem{element-isolating-cuts},
  adopting the same notation. In addition, for each $r$, let $\tilde{n}_r$ be
  the number of vertices in $\barU_r$. Note that as the $\bar{U}_r$'s are
  disjoint, we have $\sum_{r \in R} \tilde{n}_r \leq n$.

  For each $r$, we employ two different approaches to computing the
  minimum $(r,\bar{t})$-element cut.  On one hand we can apply any max
  flow algorithm in $\ectime{m_r}{n_r}$ time. As remarked above we
  have $\ectime{m_r}{n_r} \leq \apxO{m_r + n_r^{1.5}}$ by
  \cite{brand++}.  The second approach is to apply blocking flows with
  the following additional observations.  Element connectivity
  can be modeled as maximum flow in undirected graphs with edge and vertex capacities,
  which in turn can be reduced to maximum flow in edge capacitated directed graphs.
  Recall the directed graph
  representation of vertex capacities, sometimes called the ``split
  graph''. We remind the reader that in the split graph, each
  non-terminal vertex $v \in V_r \setminus \setof{r,\bar{t}}$ is split
  into two vertices -- an ``in-vertex'' $v^-$ and an ``out-vertex''
  $v^+$ -- and there is an edge $(v^-, v^+)$ with capacity equal to
  $\weight{v}$. Each edge $(u,v) \in E_r$ is replaced with an edge
  $(u^+,v^-)$ with the same capacity. From this split graph, we
  contract the edges $(v^+,\bar{t})$ for all $v \in V \cap \barU_r'$,
  which is safe because $\bar{t}$ is the sink and each edge
  $(v^+, \bar{t})$ has infinite capacity. Now, in this directed
  auxiliary graph, we have $\bigO{m_r}$ edges and $\bigO{n_r}$
  vertices. We now observe that the auxiliary vertices corresponding
  to $\barU_r' \cap V$, $\setof{v^- \where v \in \barU_r' \cap V}$,
  do not have any edges between them.  Then any $(s,\bar{t})$ path in
  this graph or in any residual graph that may arise cannot have
  consecutive auxiliary vertices from $\barU'_r$. Therefore, every
  augmenting path has length at most $2\tilde{n}_r$.  In turn, $\bigO{\tilde{n}_r}$
  iterations of blocking flows suffice to find the minimum
  $(r, \bar{t})$ cut in $G_r$, which takes
  $\bigO{m_r \log{m_r / n_r}}$ time per iteration
  \cite{goldberg-tarjan-90} and $\bigO{m_r \tilde{n}_r \log{m_r / n_r}}$ time
  overall.

  Let $\alpha > 0$ be a parameter to be determined. We have
  $\tilde{n}_r \geq \alpha n / k$ for at most $k/\alpha$ vertices $r \in
  R$. We have
  \begin{align*}
    \bigO{\sum_r \min{\ectime{m_r}{n_r}, m_r \tilde{n}_r \log{m}}} %
    &\leq                                            %
      \apxO{\sum_{r \where \tilde{n}_r \leq \alpha n / k} m_r \tilde{n}_r + \sum_{r
      \where \tilde{n}_r \geq \alpha n /k } \parof{m_r + n_r^{1.5}}}
    \\
    &\leq                        %
      \apxO{m +  \prac{\alpha}{k} m n + \prac{k}{\alpha} n^{1.5}}
    \\
    &\tago{\leq}
      \apxO{m + m^{1/2} n^{5/4}}
      =
      \apxO{m^{1/2} n^{5/4}},
  \end{align*}
  as desired.  Here, in step \tagr, we substituted
  $\alpha = k n^{1/4}/m^{1/2}$.
\end{proof}

\section{Vertex connectivity}

\newcommand{\vconn}{\kappa} \labelsection{vc} %

In this section we consider the problem of computing the vertex
connectivity in weighted and unweighted
graphs.
\setupUG Given distinct nodes $s,t \in V$ such that $st \not \in E$,
the minimum weight vertex separator between $s$ and $t$ can be
computed via flow techniques.  Recently there has been significant
improvement in the running time of vertex capacitated flow to
$\tilde{O}(m+ n^{1.5})$ \cite{brand+}. We use $\vctime{m,n}$ to denote
the complexity of computing such a separator. We let $\vconn(s,t)$
denote the weight of the separator between $s,t$ with the
understanding that $\vconn(s,t) = \infty$ if $\setof{s,t} \in E$.
Here we are interested in the minimum vertex weight separator of $G$
which can be defined as $\min_{s,t \in V, s \neq t} \vconn(s,t)$.

Let $R \subset V$ such that $R$ is an \emph{independent set} in $G$;
that is, no two vertices in $R$ share an edge. One can then define
$\vconn(R)$ to be $\min_{s,t \in R, s \neq t} \vconn(s,t)$.  We
observe that $\vconn(R)$ is the same as the element connectivity of
$R$ in the graph where $R$ is the set of terminals and $V \setminus R$
are the non-terminals and edge weights are set to $\infty$; i.e., only
vertices are allowed to be removed. We have already seen algorithms
for element connectivity, which immediately convert to isolating cut
algorithms for vertex connectivity. For instance, one can compute the
minimum isolating cut in $\apxO{m^{1+o(1)}n^{3/8}U^{1/4} + n^{1.5}}$
time with integral vertex weights between $1$ and $U$, or in
$\apxO{\sqrt{m} n^{5/4}}$ time for polynomially bounded weights.

These running times for isolating cuts do not, however, immediately
convert to running times for vertex cuts.  To obtain the minimum
vertex cut as an isolating cut, we must initialize the algorithm with
a set of vertices $R$ for which the minimum vertex cut is also an
isolating cut.  Let $(S,T)$ be opposite sides of a minimum vertex cut
$N(S) = N(T)$. Without loss of generality suppose $S$ has weight less
than or equal to $T$. We would like a set $R$ that samples exactly one
point from $S$, at least one point from $T$, and avoids $N(S)$
altogether. Even in the unweighted setting, uniform sampling is
thwarted by the fact that $N(S)$ may be much larger than $S$, and it
is difficult to hit $S$ without hitting $N(S)$ too. In the following
lemma, we observe that if we relax our problem to an
$\epsmore$-approximately minimum vertex cut, then we can sample a
useful set $R$ with reasonably good probability.
\begin{restatable}{lemma}{VCSampling}
  \labellemma{vc-sampling} Let $\eps > 0$ be fixed.  \setupG Let
  $w: V \to [1, U]$ be positive vertex weights and let
  $W = \sum_{v \in V} \weight{v}$ be the total weight.  \setupvc
  Suppose the minimum weighted degree is greater than $\epsmore
  \vc$. Then one can compute a randomized independent set
  $R \subset V$ such that the minimum vertex cut is an $R$-isolating
  set with probability at least
  \begin{math}
    \bigOmega{\parof{\eps / \log{n U}} \max{\eps, \parof{1 -
          \vc / {W}}}}
  \end{math}
\end{restatable}
\begin{proof}
  For ease of notation, let
  \begin{align*}
    \eps_0 = \max{\eps, \frac12 \parof{1 -
    \frac{\vc}{W}}}
  \end{align*}
  Let $(S,T)$ be opposite sides of the minimum vertex cut
  $N(S) = N(T)$. Without loss of generality suppose
  $w(S) \leq w(T)$ where we use the notation $w(A)$ to
  denote the total weight of
  vertices in $A$, that is, $w(A) = \sum_{v \in A} w(v)$.
  Since $N(S)$ is the minimum weight vertex separator we have $w(N(S)) = \vc$.
  We would like a independent set
  $R \subset V$ that has exactly one point from $S$, at least one
  point from $T$, and avoids $N(S)$ altogether. Then $(S,T)$ would
  isolate the lone vertex in $R \cap S$ from $R -r \subseteq T$, as
  desired. We can achieve this via a sampling procedure that we
  described below.

  First we claim that $w(S) \geq \eps \vc$. Fix an arbitrary vertex $v
  \in S$.  By assumption, $w(N(v)) \ge
  (1+\eps)\vc$. Since $N(v) \subseteq S \cup N(S)$ and $S \cap N(S) =
  \emptyset$, we have $w(S) \ge \eps \vc$.

  Let $\mu$ be any value in the range
  $[2\max{w(S), \vc}, 4\max{w(S),\vc}]$. Since
  $\max{\sumweight{S}, \vc}$ lies in the range $[1, \poly{n,U}]$, we
  can sample a value $\mu$ that lies in the above range with
  probability $\bigOmega{1/\log{n U}}$ by randomly picking a power of
  $2$ in the range $[1, \poly{n,U}]$. Once we fix $\mu$, let $R$ be a
  random subset of vertices obtained by independently sampling each
  vertex $v$ with probability $\weight{v} / \mu$. Then, as long as $R$
  has an adjacent pair of vertices, we remove one of them from $R$.
  We claim that the intial sample for $R$ has one point from $S$, no
  points from $N(S)$, and at least one point from $T$ with probability
  $\geq \bigOmega{\eps \eps_0}$.  If so, then since $S$ and $T$ are
  independent from one another, dropping vertices in the second phase
  will not remove any vertices from $S$, and retain at least one
  vertex in $T$, as desired.  Observe that the three events are
  independent.

  The probability that $R$ avoids $N(S)$ is $\prod_{v
    \in N(S)} (1- \weight{v}/\mu)$.  Since $w(N(S)) = \vc$ and $\mu
  \ge 2 \vc$, for any $v \in N(S)$, $w(v)/\mu \le 1/2$.  For $x \in (0,1/2]$ the inequality
    $(1-x) \ge e^{-2x}$ holds. Hence, $\prod_{v \in N(S)} (1-
  \weight{v}/\mu) \ge \prod_{v \in N(S)} e^{-2\weight(v)/\mu} \ge e^{-2\vc/\mu} \ge 1/e$.

  Recall that $w(S) \geq \eps \vc$ and hence the probability
  that $R$ samples exactly one
  vertex from $S$ is
  $$\sum_{v \in S} \frac{\weight{v}}{\mu} \prod_{u \in
    S-\{v\}}(1-\frac{\weight{u}}{\mu}) \ge \sum_{v \in S}
  \frac{\weight{v}}{\mu} e^{-2 (w(S)-w(v))/\mu}\ge \frac{1}{e} \sum_{v
    \in S} \frac{\weight{v}}{\mu} \ge \frac{\eps}{4e}.$$ In the
  preceding set of inequalities we used the fact that $1-x \ge
  e^{-2x}$ for $x \in [0,1/2]$ since $w(S)/\mu \le 1/2$. In the final inequality
  we used the fact that $w(S) \ge \eps \vc$ which implies that $w(S) \ge \eps \mu/4$.

  We claim that  $w(T) \ge  \eps_0 \mu/4$. Assuming the claim,
  the probability that $R$ samples at least one vertex from $T$
  is $\geq 1 - e^{-w(T)/\mu}  \ge 1-e^{\eps_0/4} =  \bigOmega{\eps_0}$.
  To see the claim, recall that $w(T) \ge w(S) \ge \eps \vc$. We also have
  $w(S) + w(T) + w(N(S)) = W$ which implies that $w(T) \ge \frac{1}{2}
  (W - \vc) \ge \frac{1}{2}(1-\frac{\vc}{W})W \ge
  \frac{1}{2}(1-\frac{\vc}{W})w(S)$. Since $\mu \le 4
  \max\{w(S),\vc\}$, we have $w(T) \ge \eps_0 \mu/4$.

  Thus, given $\mu$ lies in the range $[2\max{w(S), \vc}, 4\max{w(S),\vc}]$ which
  happens with probability $\bigOmega{1/\log{n U}}$ we have the desired sample
  $R$ with probability $\bigOmega{\eps \cdot \eps_0}$.
\end{proof}

\subsection{Approximate vertex connectivity}

By combining the isolating cut algorithms
with the sampling lemma above (for the case where no singleton already induces a good enough vertex
cut), we obtain the following approximation algorithm for vertex
connectivity.  We point out that in the running time below, the
trailing factor ($\min{1/\eps, W/(W-\vc)}$) is simply a constant
except in the relatively uninteresting setting where the minimum
weight vertex cut is almost all of the weight of the graph. In the
regime of interest, the following is a $\apxO{1/\eps}$ factor greater
than the running time to compute an isolating vertex cut.
\begin{restatable}{theorem}{ApxVC}
  \labeltheorem{apx-vertex-connectivity} Let $\eps > 0$ be fixed.
  \setupG Let $w: V \to \preals$ be positive vertex weights and let
  $W = \sum_{v \in V} \weight{v}$ be the total weight.  \setupvc Then
  a minimum vertex cut can be computed with high probability in
  \begin{math}
    \apxO{(1/\eps) \ivctime{m,n} \min{(1/\eps), W / (W - \vc)}}
  \end{math}
  randomized time, where $\ivctime{m}{n}$ is the running time to compute
  the minimum isolating vertex cut in a weighted graph of $m$ edges
  and $n$ vertices.
\end{restatable}

\begin{proof}
  Let
  \begin{align*}
    \ell = \apxO{\frac{1}{\eps} \min{\frac{1}{\eps},
    \frac{W}{W - \vc}}}
  \end{align*}
  The algorithm first repeats the following subroutine
  $\bigO{\ell}$ times. This subroutine first generates an set $R \subset V$
  by \reflemma{vc-sampling}, and then it computes a minimum $R$-isolating cut.
  It compares the $\ell$ isolating cuts generated above with the
  singleton cuts in the graph, returning the minimum overall.

  We argue that the algorithm returns a $\epsmore$-approximate
  minimum weight cut with high probability by the following simple
  analysis. In one case, some singleton cut is an approximate minimum
  cut, in which case the algorithm always succeeds. In the second
  case, the minimum weighted degree is at least an
  $\epsmore$-multiplicative factor greater than the vertex
  connectivity. In that case, the minimum weight vertex cut is a
  minimum $R$-isolating cut for at least one of the random sets $R$
  with high probability, in which case we return the minimum weight
  vertex cut.
\end{proof}

We briefly compare our bound above to previous work. As mentioned
previously Henzinger, Rao and Gabow \cite{hgr} obtain a randomized
algorithm that gives the exact vertex connectivity in $\apxO{mn}$ time
for weighted graphs.  We obtain a $(1+\eps)$-approximation in
$\apxO{m \sqrt{n}/\eps}$ time or in $\apxO{m^{1/2}n^{5/4}/\eps}$ time;
other bounds are outlined in previous subsection. We are thus able to
obtain substantially faster algorithm if we settle for a small
approximation. There have been past works on approximation for vertex
connectivity but as far as we know they have been limited to
unweighted graphs. Henzinger \cite{h-97} obtained a $2$-approximation
in $O(n^2\min(\sqrt{n},\vc))$. Forster et al.\ obtained a
$(1+\eps)$-approximation in randomized time $\apxO{m+ n\vc^2/\eps}$
which is near-linear for small connectivity, and combining various
other results they improve upon Henzinger's result. We refer the
reader to \cite{fnsyy} for detailed bounds. Our running times are
useful for the larger connectivity regime and we can obtain improved
bounds in various other regimes of interest. We leave a more detailed
comparison to a future version of the paper.

\subsection{Exact vertex connectivity}

Now, for integral weights, the approximation algorithm above gives the
following exact algorithm for vertex connectivity by suitable choice
of $\eps$.  Again we highlight that in the running time below, the
trailing factor ($\min{\vc, \frac{W}{W -\vc}}$) is simply a constant
except in the relatively uninteresting setting where $\vc$ is almost
$\sum_{v \in V} \weight{v}$, in which case the remaining factors of
$\bigO{\vc \ivctime{m}{n}}$ are not as compelling anyway.
\begin{restatable}{theorem}{ExactVC}
  \labeltheorem{exact-vc} \setupIntUGvc Then the minimum vertex cut
  can be computed with high probability in
  \begin{math}
    \apxO{\vc \ivctime{m,n} \min{\vc, W / (W - \vc)}}
  \end{math}
  randomized time, where $\ivctime{m}{n}$ is the running time to compute
  the minimum isolating vertex cut in a weighted graph of $m$ edges
  and $n$ vertices.
\end{restatable}
\begin{proof}
  For integral capacities, a $(1 + 1/(\vc + 1))$-approximation is an
  exact solution. Thus the result follows from
  \reftheorem{apx-vertex-connectivity}.
\end{proof}

For the unweighted case, combining the above with sparsification
\cite{nagamochi-ibaraki} gives the following.

\begin{restatable}{corollary}{ExactVCSparse}
  \labelcorollary{exact-vc-sparse} Let $G=(V,E)$ be a simple
  unweighted graph. Then the minimum vertex cut can be computed with
  high probability in
  \begin{math}
    \apxO{m + \vc \ivctime{n \vc,n} \min{\vc, n / (n - \vc)}}
  \end{math}
  randomized time, where $\ivctime{m}{n}$ is the running time to compute
  the minimum isolating vertex cut in a graph of $m$ edges
  and $n$ vertices.
\end{restatable}
\begin{proof}
  For unweighted graphs we can assume we know $\vc$ (via exponential
  search which adds an additional $O(\log \vc)$ overhead).  We apply the
  well-known linear-time sparsification algorithm of Nagamochi and
  Ibaraki \cite{nagamochi-ibaraki} to reduce the number of edges to
  $O(n\kappa)$ and then run the algorithm in the preceding theorem
  on the sparsified graph which gives the claimed bound.
\end{proof}

\begin{table}[htb]
  \everymath{\displaystyle}
  \renewcommand\tabularxcolumn[1]{m{#1}}%
  \centering
  \begin{tabularx}{\textwidth}{| >{\emstrut{1.5}{1}}c | X |}
    \hline $\bigO{n^2 \vctime{\vc n,n}}$ & Combines trivial algorithm
    with
    sparsification \cite{nagamochi-ibaraki}. \\
    \hline $\bigO{n \vctime{\vc n,n}}$ & Combines randomized trivial
    algorithm with sparsification
    \cite{nagamochi-ibaraki}. $\vc \leq .99n$ \\
    \hline %
    $\bigO{n^{\omega} + n \vc^{\omega}}$. & \cite{llw}.
    \\
    \hline %
    $\apxO{\vc n^2}$ & \cite{hgr}. Randomized. \\
    \hline %
    $\bigO{\min{n^{3/4}, \vc^{3/2}} \vc^2 n + \vc n^2}$ &
    \cite{gabow-06}.
    \\
    \hline $\apxO{m + \vc^{7/3} n^{4/3}}$ &
    {\cite{nsy-19}. Randomized.} %
    \\ \hline %
    $\apxO{m + n \vc^3}$ & {\cite{fnsyy}. Randomized.}  \\ \hline %
    $\apxO{m + \vc^{7/3} n^{4/3}}$ &
    \refcorollary{exact-vc-sparse}. Randomized.  $\vc \leq .99n$
    \\
    \hline $\apxO{m + \vc^2 n^{11/8+o(1)} + \vc n^{3/2}}$ &
    \refcorollary{exact-vc-sparse}. Randomized.
    $\vc \leq .99n$
    \\
    \hline $\apxO{m + \vc^{1.5} n^{7/4}}$ &
    \refcorollary{exact-vc-sparse}. Randomized. $\vc \leq .99n$.
    \\
    \hline %
  \end{tabularx}

  \caption{A table of running times for finding the minimum vertex cut
    in an \emph{unweighted} and undirected graph.  $\vctime{m,n}$
    denotes the running time of computing $(s,t)$-vertex
    connectivity. $\ectime{m,n}$ denotes the running time computing
    $(s,t)$-edge connectivity.  See also \cite[Section
    15.2a]{schrijver}. \labeltable{unweighted-vc-running-times}}

  \graybar
\end{table}


The reduction from exact vertex connectivity to isolating vertex cut
above, mixed with the algorithms for isolating vertex cuts, and optionally including the
sparsification step from \refcorollary{exact-vc-sparse}, produces a
number of new running times that are optimal for different ranges of
$\vc$. In general, the running times obtained here have a lower
dependence on $\vc$ then other algorithms for vertex connectivity with
a $\poly{\vc}$ dependence (which is common for the unweighted
setting), so the running times here are particularly good for moderate
to large $\vc$. For a more detailed comparison between the literature
and new running times for the unweighted setting (where we restrict to
unweighted for simplicity), see
\reftable{unweighted-vc-running-times}.


\section{Hypergraph Connectivity}
\labelsection{hypergraphs} Let $H=(V,E)$ be a weighted hypergraph and
let $R \subseteq V$.  The cut function of a hypergraphs is symmetric
and submodular.  Given disjoint sets $S, T \subset V$ the minimum
$S$-$T$ cut in $H$ can be computed in $\ectime{p,m+n}$ time
via standard reductions\footnote{One can also reduce to computing $s$-$t$ cut
  in a vertex capacitated undirected graph with $p$ edges and $m+n$ nodes,
  although there does not seem to be any particular advantage with current
  running time bounds for $\ectime{p,m+n}$.}. We can use \reflemma{isolating-cut-partition}
and \refcorollary{symsubmod-mincut} to understand the running time to
compute $R$-connectivity in $H$. Up to logarithmic factors it suffices
to estimate the time to find $R$-isolating cuts. Recall that the
running time consists of two parts. The first part is $O(\log |R|)$
calls to $S$-$T$ cut problem in $H$. After this we have the following
situation. For each $r \in R$ we obtain a set $U_r \subset V$ such
that $r \in R$ and $U_r \cap (R-r) = \emptyset$.  Furthermore the sets
$U_r$ over $r \in R$ are pairwise disjoint.  For each $r$ the goal is
to find a set $Y_r \subseteq U_r$ with minimum $w(\delta(Y_r))$ where
$\delta(Y_r)$ is set of hyperedges crossing $Y_r$. Let $n_r =
|U_r|$. We can compute $Y_r$ by solving a cut problem in an auxiliary
hypergraph $G_r$ on $n_r+1$ vertices obtained by shrinking
$V \setminus U_r$ into a single vertex. Let $p_r$ be the total size of
the hyperedges in $G_r$. It is not hard to see that
$\sum_{r \in R} p_r = \bigO{p}$. Thus each cut problem in $G_r$ can be
computed in either $\ectime{p_r,m+n_r+1}$. This implies the following.

\begin{theorem}
  The minimum isolating cuts over a set of vertices $R$ of size
  $k = \sizeof{R}$ in a hypergraph with $m$ edges, $n$ vertices, and
  total size $p$ can be computed
  \begin{align*}
    \apxO{\ectime{p, m + n} +
    \max[n_1,\dots,n_k,p_1,\dots,p_k]{\sum_{i=1}^k \ectime{p_i}{m + n_i} \where
    n_1 + \cdots + n_k \leq n \andcomma p_1 + \cdots + p_k \leq 2p}}
  \end{align*}
  time with high probability.
\end{theorem}

In particular $\ectime{p}{m+n}$ is $\tilde{O}(p\sqrt{m+n} \log U)$
\cite{lee-sidford} and for unweighted case we have
$\ectime{p}{m+n} = \tilde{O}(p^{4/3})$ \cite{liu-sidford-20-a}. We can
obtain two other run times for hypergraphs that provide different
tradeoffs. These are obtained by more carefully solving the second
part of the isolating cut framework, and transfer ideas from vertex
connectivity to hypergraphs.

\begin{enumerate}
\item $\sqrt{p n (m+n)^{1.5}}$.
\item
  $\apxO{p (m+n)^{\frac{3 \alpha}{2 (1 + \alpha)}} \beta^{\frac{1}{1 +
        \alpha}}}$ for any $\alpha,\beta$ where
  $\ectime{m}{n} \leq m^{1 + \alpha} \beta$ (e.g.,
  \cite{liu-sidford-20-b} gives
  $\ectime{m}{n} \leq \apxO{m^{4/3} U^{1/3}}$, which we interpret as
  $\alpha = 1/3$ and $\beta = U^{1/3}$).
\end{enumerate}
We sketch the proofs of theorems that obtain the preceding bounds.

\begin{theorem}
  The minimum isolating cut in a hypergraph can be
  computed in
  \begin{align*}
    \apxO{\sqrt{p n (m+n)^{1.5}}}
  \end{align*}
  randomized time.
\end{theorem}
\begin{proof}
  The approach is similar to the algorithm for element isolating cuts
that had a running time of $\apxO{\sqrt{m} n^{5/4}}$, and we restrict
ourselves to a sketch. Let $k = \sizeof{R}$. Note that a single
$(S,T)$-hypergraph cut can be computed in $\apxO{p + (m+n)^{1.5}}$ by
reduction to edge capacitated flow, or in $\apxO{p n}$ time by
blocking flows (where we observe that every augmenting path has length
$\leq \bigO{n}$).  With $\bigO{\log k}$ calls to such a subroutine to
$(S,T)$-hypergraph cut (the first phase), the bisubmodular crossing
framework produces vertex disjoint sets $U_r \subset V$ such that
$r \in U_r$ and the minimum $(r,R-r)$-hypergraph cut is induced by a
subset of $U_r$. Each $(r,R-r)$-cut problem can be solved by an
$(s,t)$-hypergraph cut problem where we contract all vertices in
$V - U_r$ to a sink node $t$, and use $r$ as the source $r$. Let
$n_r = \sizeof{U_r} + 1$ be the number of nodes in this auxiliary
graph, let $m_r$ be the number of edges, and let $p_r$ be the total
size. This sub-problem can be solved in either
$\apxO{p_r + (m_r+n_r)^{1.5}}$ by edge capacitated flow or
$\apxO{p_r n_r}$ time by running blocking flows (as observed
above). We do not have a good bound on $\sum_r m_r$ since a hyperedge
$e$ may intersect many sets $U_r, r \in R$, hence we simply use the
trivial bound that $m_r \le m$ for all $r$.

Let $\alpha > 0$ be a parameter to be determined. We run the
blocking flow approach if $n_r \leq \alpha n / k$ and the edge
capacitated flow approach if $n_r > \alpha n / k$.  The total time
for the blocking flow computations is $\apxO( (\sum_r p_r)\alpha
n/k) = \apxO(\alpha n p/k)$ since $\sum_r p_r = O(p)$.  Since
$\sum_r n_r \leq n$, we run the edge capacitated flow approach for
at most $(k/\alpha)$ choices of $r$ and hence the total time for
these computation is $\apxO(p + (m+n)^{1.5} k/\alpha)$.
By choosing $\alpha = k \sqrt{(m+n)^{1.5}/pn}$, this leads
to a total running time of $\apxO(\sqrt{p n (m+n)^{1.5}})$ for the second
phase (noting that $p \leq m n$).

The overall running time over the two phases is therefore
\begin{align*}
  \apxO{\sqrt{p n (m+n)^{1.5}} + \min{p n, p +  (m+n)^{1.5}}} \leq
  \apxO{\sqrt{p n (m+n)^{1.5}}}.
\end{align*}
\end{proof}

\begin{theorem}
  Suppose $\ectime{m}{n} \leq \apxO{m^{1+\alpha} \beta}$ for fixed
  $\alpha, \beta > 0$. Then minimum isolating cuts can be computed in
  $\apxO{p (m+n)^{\frac{3 \alpha}{2 (1 + \alpha)}} \beta^{\frac{1}{1 +
        \alpha}} + p + (m+n)^{1.5}}$.
\end{theorem}
\begin{proof}
  The approach is similar to the algorithm for element isolating cuts
  that obtained a running time of
  $\apxO{m^{1+o(1)}n^{3/8} U^{1/4 } + n^{1.5}}$, and we restrict
  ourselves to a sketch. Let $k = \sizeof{R}$. Note that a single
  $(S,T)$-hypergraph cut can be computed in $\apxO{p + (m+n)^{1.5}}$
  by reduction to edge capacitated flow.  With $\bigO{\log k}$ calls
  to a subroutine to $(S,T)$-hypergraph cut, the bisubmodular crossing
  framework produces vertex disjoint sets $U_r \subset V$ such that
  $r \in U_r$ and the minimum $(r,R-r)$-hypergraph cut is induced by a
  subset of $U_r$. Each $(r,R-r)$-cut problem can be solved by an
  $(s,t)$-hypergraph cut problem where we contract all vertices in
  $V - U_r$ to a sink node $t$, and use $r$ as the source $r$. Let
  $n_r = \sizeof{U_r} + 1$ be the number of nodes in this auxiliary
  graph, let $m_r$ be the number of edges, and let $p_r$ be the total
  size. This sub-problem can be solved in either
  $\apxO{p_r + (m_r+n_r)^{1.5}}$ by reduction to edge capacitated flow
  or $\apxO{p_r^{1+\alpha} \beta}$ time by reducing to edge
  capacitated flow.

  Let $\ell > 0$ be a parameter to be determined. For each subproblem
  for $r \in R$, we run edge capacitated flow algorithm if
  $p_r \leq \ell$ and the edge capacitated flow algorithm if
  $p_r \geq \ell$. Since $\sum_r p_r \leq p$, we run the edge
  capacitated flow algorithm for at most $p/\ell$ choices of $r$. For
  $\ectime{m,n} = m^{1+\alpha} \beta$, the total running time on edge
  capacitated flow is maximized by $(p/\ell)$ subgraphs with $\ell$
  edges each. This gives a total running time of
  $(p / \ell) \ell^{1+\alpha} \beta$ for this class of
  subproblems. Choosing $\ell$ to balance terms leads to the desired
  running time.
\end{proof}


\printbibliography

\appendix

\section{Isolating Cuts, Bisubmodular Functions, and
  Lattices (expanded)}
\labelsection{bisubmod}

This is an expanded version of \refsection{bisubmod-short} regarding
the bisubmodular function and uncrossing lattice framework. This
version more formally structures the definitions, includes missing
proofs, and has several examples to make the description accessible
and self-contained.  A reader familiar with the background can skip
several parts.

We note that when the same lemma or theorem from
\refsection{bisubmod-short} appears here, we assign it the same number
as given in \refsection{bisubmod-short}.

\subsection{Lattices over Set Pairs}

\begin{restatable}{definition}{SetPair}
  Let $V$ be a finite set of elements. An ordered pair $(A,B) \in 2^V
  \times 2^V$ is a \defterm{set-pair}  over $V$.
\end{restatable}

\begin{restatable}{definition}{CrossingLattice}
  Let $V$ be a finite set of elements, and let $\defV$ be a family of
  set-pairs over $V$. We say that $\V$ is a \defterm{crossing
    lattice}\footnote{This notion is analogous to the definition of a
    crossing family of sets.}  over $V$ if it is closed under the
  following two operators.
  \begin{gather*}
    (X_1,Y_1) \lor (X_2,Y_2) = (X_1 \cup X_2, Y_1 \cap Y_2). \\
    (X_1,Y_1) \land (X_2,Y_2) = (X_1 \cap X_2, Y_1 \cup Y_2).
  \end{gather*}
  If $\V$ is closed under these operations, then $\V$ is a lattice
  under the partial order
  \begin{align*}
    (X_1,Y_1) \preceq (X_2,Y_2) \iff X_1 \subseteq X_2 \andcomma Y_2
    \subseteq Y_1.
  \end{align*}
  The binary operator $\lor$ returns the unique least upper bound of
  its arguments (a.k.a.\ the \defterm{meet}) and the binary operator
  $\land$ returns the unique greatest lower bound of its arguments
  (a.k.a.\ the \defterm{join}).
\end{restatable}

\begin{restatable}{definition}{SymmetricLattice}
  For a pair of sets $(X,Y) \in \subsetsof{V} \times \subsetsof{V}$,
  the \defterm{transpose} of $(X,Y)$, denoted $(X,Y)^T$, is the
  reversed pair of sets $(X,Y)^T \defeq (Y,X)$. A crossing lattice
  $\defV$ is \defterm{symmetric} if is closed under taking the
  transpose.
\end{restatable}

We have the following identities relating the transpose with the
lattice operations $\lor$ and $\land$.
\begin{restatable}{lemma}{SymmetricLatticeOperations}
  Let $\defV$ be a symmetric crossing lattice and let $\X, \Y \in
  \V$. Then we have the following.
  \begin{align*}
    \parof{\X^T}^T &= \X. \\
    \parof{\X \lor \Y}^T &= \X^T \land \Y^T. \\
    \parof{\X \land \Y}^T &= \X^T \lor \Y^T.
  \end{align*}
\end{restatable}

\begin{restatable}{definition}{PairwiseDisjoint}
  A crossing lattice $\defV$ is \defterm{pairwise disjoint} if
  $X \cap Y = \emptyset$ for all $(X,Y) \in \V$.
\end{restatable}

\begin{restatable}{example}{Trivial}
  \labelexample{all-sets} Let $V$ be a set. The family of all
  set-pairs, $\V = \subsetsof{V} \times \subsetsof{V}$, is closed
  under $\lor$ and $\land$, and symmetric.
\end{restatable}

\begin{restatable}{example}{DisjointSets}
  \labelexample{disjoint-sets}
  Let $V$ be a set. The family of all disjoint set-pairs,
  \begin{align*}
    \V = \setof{(X,Y) \where X, Y \subseteq V \text{ and } X \cap Y = \emptyset},
  \end{align*}
  is closed under $\lor$ and $\land$, symmetric, and pairwise disjoint.
\end{restatable}

\begin{restatable}{example}{Bipartitions}
  \labelexample{bipartitions}
  Let $V$ be a set. The family of all bi-partitions of $V$,
  \begin{align*}
    \V = \setof{(X,V \setminus X) \where X \subseteq V}
  \end{align*}
  is closed under $\lor$ and $\land$, symmetric, and pairwise
  disjoint.
\end{restatable}

\begin{restatable}{example}{VertexCutSides}
  \labelexample{vertex-cut-sides}
  Let $G = (V,E)$ be an undirected graph. The family
  \begin{align*}
    \V = \setof{                                 %
    (X,Y) \in \subsetsof{V} \times \subsetsof{V} %
    \where                                                   %
    Y \cap \parof{X \cup N(X)} = \emptyset %
    \andcomma                            %
    X \cap \parof{Y\cup N(Y)} = \emptyset %
    },
  \end{align*}
  which describes pairs of disjoint vertex sets with no edge between
  them, forms a crossing lattice over $V$ that is symmetric and
  pairwise disjoint. Indeed, suppose $(X_1,Y_1)$ be a pair of disjoint
  vertex sets with no edge between them. Let $(X_2,Y_2)$ be a second
  pair of disjoint vertex sets with no edge between them. Then
  $X_1 \cup X_2$ and $Y_1 \cap Y_2$ are vertex disjoint, and have no
  edge between them. Thus $(X_1,Y_1) \lor (X_2,Y_2) \in \V$. By
  symmetry, $(X_1,Y_1) \land (X_2,Y_2) \in \V$ as well.
\end{restatable}

\begin{restatable}{example}{VertexCutSidesR}
  \labelexample{vertex-cut-sides-R}
  Let $G = (V,E)$ be an undirected graph, and let $R \subset V$ be a
  fixed set of vertices. The family
  \begin{align*}
    \V = \setof{(X,Y) \in \subsetsof{V} \times \subsetsof{V} \where Y
    \cap \parof{X \cup N(X)} = \emptyset \andcomma X \cap \parof{Y
    \cup N(X)} = \emptyset \andcomma R \subseteq X \cup Y},
  \end{align*}
  which describes pairs of disjoint vertex sets that (a) have no edge
  between them and (b) cover $R$, forms a crossing lattice over $V$ that
  is symmetric and pairwise disjoint. Indeed, it is easy to see that
  $\V$ is symmetric and pairwise disjoint. Let
  $(X_1,Y_1), (X_2,Y_2) \in \V$. As discussed in
  \refexample{vertex-cut-sides} above, $X_1 \cup X_2$ and
  $Y_1 \cap Y_2$ are disjoint and have no edge between
  them. Additionally, if $R \subseteq X_1 \cup Y_1$ and
  $R \subseteq X_2 \cup Y_2$, then
  $R \subseteq X_1 \cup X_2 \cup (Y_1 \cap Y_2)$. Thus
  $(X_1,Y_1) \lor (X_2,Y_2) \in \V$. By symmetry
  $(X_1,Y_1) \land (X_2,Y_2) \in \V$ as well.
\end{restatable}

\subsection{Cuts in crossing lattices}

In this section, we define an abstract, lattice-based notion of cuts
that unifies the various different families of cuts of interest in
graphs.

\begin{restatable}{definition}{}
  Let $V$ be a set.  For two set-pairs
  $\S = (S,T) \in \subsetsof{V} \times \subsetsof{V}$
  and $\X = (X,Y) \in \subsetsof{V} \times \subsetsof{V}$, we denote
  \begin{align*}
    \S \subseteq \X \defiff S \subseteq X \andcomma T \subseteq Y.
  \end{align*}
  If $\S \subseteq \X$, then we say that $\X$ \defterm{cuts} $\S$ or
  that $\X$ is an \defterm{$\S$-cut}. If $\V$ is a crossing lattice over
  $V$, $R \subset V$ is a subset, and $\R$ is a crossing lattice over
  $R$, then we say that $\V$ \defterm{separates} $\R$ if for every
  $\S \in \R$, there is an $\S$-cut $\X \in \V$.
\end{restatable}

\begin{restatable}{example}{}
  Let $G = (V,E)$ be a graph, let
  $\V= \setof{(X, V \setminus X) \where X \subseteq V}$ be the lattice of
  bi-partitions over $V$ (\refexample{bipartitions}) and let
  $\R = \setof{(S, T) \where S,T \subseteq R \andcomma S \cap T =
    \emptyset}$ be the lattice of disjoint sets of $R$
  (\refexample{disjoint-sets}). For $\S = (S,T) \in \R$, and
  $\X = (X,Y) \in \V$, $\X$ is an $\S$-cut iff $(X,Y)$ are opposite
  sides of an $(S,T)$-edge-cut in the usual graphical sense. The
  crossing lattice $\V$ separates $\R$.
\end{restatable}

\begin{restatable}{example}{}
  Let $G = (V,E)$ be a graph, and let $R \subset V$ be an independent
  set of vertices. Let $\V$ be the lattice of opposite sides of vertex
  cuts excluding $R$ (\refexample{vertex-cut-sides-R}).  For
  $\S = (S,T) \in \R$, and $\X = (X,Y) \in \V$, $\X$ is an $\S$-cut if
  $S \subseteq X$ and $T \subseteq Y$. By definition of $\V$, $(X,Y)$
  have no edge between them, and $V \setminus (X \cup Y)$ gives a
  vertex cut between $S$ and $T$. Conversely, given $S = (S,T) \in \R$
  and a vertex cut $W \subset V$ separating $S$ from $T$, let $X$ be
  the union of components in $G \setminus W$ containing vertices in
  $S$ and let $Y$ be the union of components in $G \setminus W$
  that do not contain any vertices in $S$ (hence they contain all
  vertices in $T$). Then $X$ and $Y$ are disjoint vertex
  sets separated by $W$, so $(X,Y) \in \V$. Thus $\V$ separates $\R$.
\end{restatable}

\LatticeCuts*

\begin{figure}
  \centering
  \includegraphics{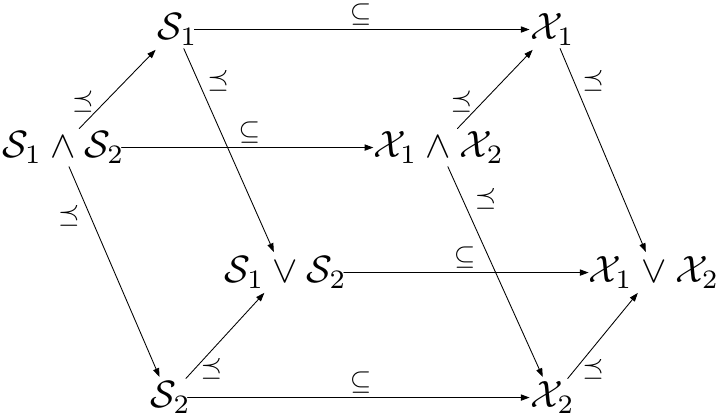}
  \caption{A diagram depicting \reflemma{lattice-cuts}, which asserts
    that cuts are preserved by the lattice operations $\lor$ and
    $\land$.\labelfigure{lattice-cuts}}
  \graybar
\end{figure}

\begin{proof}
  The proof is by direct inspection, and we include the details for
  the sake of completeness. Let $\S_i = (S_i,T_i)$ and
  $\X_i = (X_i,Y_i)$ for $i = 1,2$. For $i=1,2$, $\S_i \subseteq \X_i$
  means that $S_i \subseteq X_i$ and $T_i \subseteq Y_i$ for $i = 1,2$
  by definition of $\subseteq$. Therefore
  $S_1 \cap S_2 \subseteq X_1 \cap X_2$, $S_1 \cup S_2 \subseteq X_1
  \cup X_2$, $T_1 \cup T_2 \subseteq Y_1 \cup Y_2$, and $T_1 \cap T_2
  \subseteq Y_1 \cap Y_2$. Thus
  \begin{align*}
    \S_1 \land \S_2 = \parof{S_1 \cap S_2, T_1 \cup T_2} \subseteq
    (X_1 \cap X_2, Y_1 \cup Y_2) = \X_1 \land \X_2,
  \end{align*}
  and
  \begin{align*}
    \S_1 \lor \S_2 = \parof{S_1 \cup S_2, T_1 \cap T_2} \subseteq %
    (X_1 \cup X_2, Y_1 \cap Y_2) = \X_1 \lor \X_2,
  \end{align*}
  as desired.
\end{proof}

See \reffigure{lattice-cuts} for a diagram of \reflemma{lattice-cuts}.

\subsection{Submodular functions over lattices and Bisubmodular functions}

\begin{restatable}{definition}{}
  Let $\V$ be a lattice. A real-valued function $f: \V \to \reals$ is
  \defterm{submodular} if for all $\X,\Y \in \V$,
  \begin{align*}
    \f{\X} + \f{\Y} \geq \f{\X \lor \Y} + \f{\X \land \Y}.
  \end{align*}
\end{restatable}

\begin{restatable}{example}{}
  \emph{Bisubmodular functions} can be interpreted as submodular
  functions over particular crossing lattices. There are at least two
  definitions of bisubmodular function in the literature. These
  definitions are similar and we discuss both.

  In one definition (e.g., in \cite{schrijver}), a function
  $f: \subsetsof{V} \times \subsetsof{V} \to \reals$ is called
  \emph{bisubmodular} if for all $X_1,Y_1,X_2,Y_2 \subseteq V$, we
  have
  \begin{align*}
    \f{X_1,Y_1} + \f{X_2,Y_2} \geq %
    \f{X_1 \cup X_2, Y_1 \cap Y_2} + \f{X_1 \cap X_2, Y_1 \cup Y_2}.
    \labelthisequation{bisubmodular}
  \end{align*}  A bisubmodular function
  $f: \subsetsof{V} \times \subsetsof{V} \to \reals$ is submodular
  over the crossing lattice of all set-pairs,
  $\V = \subsetsof{V} \times \subsetsof{V}$ (\refexample{all-sets}).

 Another definition (e.g.,
  \cite{bouchet-87,afn-96,af-96,fujishige-iwata}) of a bisubmodular
  function $f$ is that $f(X_1,Y_1)$ is only defined for disjoint sets
  $X_1$ and $Y_1$, and otherwise satisfies inequality
  \refequation{bisubmodular} for these inputs. In this version, $f$ is
  bisubmodular iff it is a submodular function over the lattice of
  disjoint sets,
  $\V = \setof{(X,Y) \where X,Y \subseteq V \andcomma X \cap Y =
    \emptyset}$ (\refexample{disjoint-sets}).
\end{restatable}

We now define the notion of a symmetric submodular function \emph{over
  a crossing lattice}. This is a different definition then for
symmetric submodular \emph{set} functions and (by
\refexample{submod-to-bisubmod} below) the lattice-based definition
generalizes the (more standard) set-based definition.

\begin{restatable}{definition}{}
  Let $\defV$ be a symmetric crossing lattice. A function $f: \V \to
  \reals$ is \defterm{symmetric} if for all $\X \in \V$,
  \begin{math}
    \f{\X} = \f{\X^{T}}.
  \end{math}
\end{restatable}

\begin{restatable}{example}{}
  \labelexample{submod-to-bisubmod}
  Any symmetric submodular set function $\deff$ can be interpreted as a
  symmetric submodular function $h: \V \to \reals$ over the crossing
  lattice of bipartitions,
  $\V = \setof{(X, V \setminus X ) \where X \subseteq V}$. Here $h$ is
  defined by
  \begin{align*}
    h(X, V \setminus X) = \f{X} = \f{V \setminus X}.
  \end{align*}
\end{restatable}

\begin{restatable}{example}{}
  \labelexample{vertex-cut-weight} Let $G = (V,E)$ be an undirected
  graph with non-negative vertex weights $\weight: V \to \nnreals$. Let $\V$
  be the family of disjoint set-pairs with no edge between them
  (\refexample{vertex-cut-sides}). We define a function $\f: \V \to \reals$
  by
  \begin{align*}
    \f(X,Y) = \sum_{v \in V \setminus \parof{X \cup Y}} \weight{v}.
  \end{align*}
  It is easy to see that $f$ is submodular (and in fact, $f$ is
  modular).
\end{restatable}

\subsection{Minimal and minimum submodular cuts}

\begin{figure}
  \centering                    %
  \includegraphics{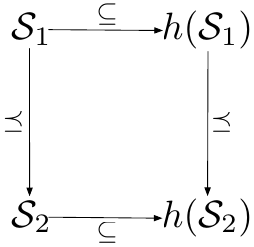} %
  \caption{A diagram depicting \reflemma{lattice-min-cuts}, which
    asserts that $\preceq$-minimal, $f$-minimum cuts preserve
    order. Here $h(\S)$ is the $\preceq$-minimal, $\f$-minimum
    $\S$-cut, which is well-defined. \labelfigure{lattice-min-cuts}}
  \graybar
\end{figure}

The following lemma outlines a key relationship between the sets of
terminals being separated and \emph{minimal} minimum cuts that
separate them. See \reffigure{lattice-min-cuts} for a diagrammatic
description of the following lemma.

\LatticeMinCuts*

\begin{proof}
  We first prove that $h$ is well-defined. It suffices to show that
  for any $\S \in \R$, the meet of two minimum $\S$-cuts in $\V$ is a
  minimum $\S$-cut. Let $\X_1, \X_2 \in \V$ be two minimum $\S$-cuts.
  We have
  \begin{align*}
    \f{\X_1} + \f{\X_2} \tago{\geq} \f{\X_1 \lor \X_2} + \f{\X_1 \land
    \X_2}
    \tago{\geq} \f{\X_1} + \f{\X_1 \land \X_2}.
  \end{align*}
  Here \tagr is by (generalized) submodularity. \tagr is because
  $X_1 \lor \X_2$ is also an $\S$-cut and $\X_1$ is a minimum
  $\S$-cut. Canceling like terms, we have
  $\f{\X_1 \land \X_2} \leq \f{\X_2}$. Moreover, $\X_1 \land \X_2$ is
  an $\S$-cut and $\X_2$ is a minimum $\S$-cut, so $\X_1 \land \X_2$
  is also a minimum $\S$-cut.

  We now show that $h$ preserves order.  Let $\S_1,\S_2 \in \R$,
  $\X_1 = h(\S_1) \in \V$, and let $\X_2 = h(\S_2) \in \V$. We want to
  show that if $\S_1 \preceq \S_2$, then $\X_1 \preceq \X_2$.
  Suppose by contradiction that $S_1 \preceq S_2$ and
  $\X_1 \not\preceq \X_2$.
  \begin{align*}
    \f{\X_1} + \f{\X_2}         %
    \tago{\geq}                 %
    \f{\X_1 \lor \X_2} + \f{\X_1 \land \X_2} %
    \tago{\geq}                             %
    \f{\X_2} + \f{\X_1 \land \X_2} %
    \tago{>}                     %
    \f{\X_2} + \f{\X_1},           %
  \end{align*}
  a contradiction.  Here \tagr is by submodularity. \tagr is because
  $\X_1 \lor \X_2$ is also an $\S_2$-cut, and $\X_2$ is an $f$-minimum
  $\S_2$-cut. \tagr is for the following reasons. First,
  $\X_1 \land \X_2$ is also an $\S_1$-cut.  Second, if
  $\X_1 \not\preceq \X_2$, then $\X_1 \land \X_2 \prec \X_1$
  (strictly). Since $\X_1$ is $\preceq$-minimal, $\X_1 \land \X_2$
  cannot be a minimum $\S_1$-cut, and
  $\f{\X_1 \land \X_2} > \f{\X_1}$.
\end{proof}

The following is a particularly convenient form of
\reflemma{lattice-min-cuts}, and the one we will actually apply in the
subsequent subsection. A diagram depicting the following lemma is given
in \reffigure{uncrossing-min-cuts}.

\UncrossingMinCuts*

\begin{figure}
  \centering
  \includegraphics{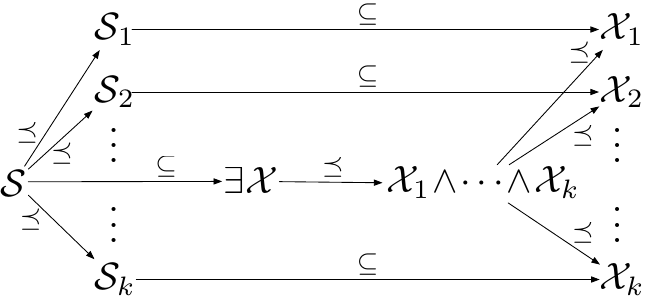}
  \caption{A diagram depicting \reflemma{uncrossing-min-cuts}. Here,
    for $i = 1,\dots,k$, $\X_i$ is a minimum $\S_i$ cut.  $\X$ is a
    minimum $\S$-cut whose existence is asserted by
    \reflemma{uncrossing-min-cuts}. \labelfigure{uncrossing-min-cuts}}

  \graybar
\end{figure}

\begin{proof}
  Let $\X$ be the $\preceq$-minimum, $f$-minimum $\S$-cut. By
  \reflemma{lattice-min-cuts}, $\X$ is well-defined.  Let
  $\T = \S_1 \land \cdots \land \S_k$.  If $\S \preceq \S_i$ for all
  $i \in [k]$, then since $\T$ is the join of $\S_1,\dots,\S_k$, we
  have $\S \preceq \T$. Let $\Y$ be the $\preceq$-minimum, $f$-minimum
  $\T$-cut. By \reflemma{lattice-min-cuts}, since $\S \preceq \T$,
  $\X$ is the $\preceq$-minimum, $f$-minimum $\S$-cut, and $\Y$ is an
  $f$-minimum $\T$-cut, we have $\X \preceq \Y$. Again by
  \reflemma{lattice-min-cuts}, since $\Y \preceq \X_i$ for all $i$, we
  have that $\Y \preceq \X_1 \land \cdots \land \X_k$. By transitivity
  we have $\X \preceq \X_1 \land \cdots \land \X_k$, as desired.
\end{proof}

\subsection{Isolating symmetric submodular cuts over lattices}

We now come to the issue of computing isolating cuts. We formalize
this as follows. Let $V$ be a set and $R \subset V$. Let $\defV$ be a
symmetric and \emph{pairwise disjoint} crossing lattice over $V$ and
let $\defR$ be the symmetric and \emph{pairwise disjoint} crossing
lattice over $R$ consisting of all partitions of $R$; i.e.,
\begin{math}
  \R = \setof{(S,T) \where S \cup T = R \andcomma S \cap T =
    \emptyset}.
\end{math}
Let $f: \V \to \reals$ be a symmetric bisubmodular function.  For
each $r \in R$ we wish to find an $f$-minimum cut $\Y_r$ for the set-pair
$(\{r\}, R-\{r\})$ (which we abbreviate as $(r,R-r)$ for notational
simplicity). The main property that leads to efficiency
is captured by the next lemma.


\IsolatingCutPartition*


\begin{proof}
  Let $k = \logup{\sizeof{R}}$. One can choose $k$ partitions
  $\S_1,\dots,\S_k \in \R$ such that every pair of elements in $R$ is
  separated by at least one partition\footnote{We briefly describe the
    simple construction from \cite{li-panigrahi} for the sake of
    completeness. Enumerate the vertices in $R$ from $1$ to
    $\sizeof{R}$ and consider the binary representations of these
    indices. For each $i$, let $S_i$ be the set of vertices whose
    $i$th bit is 0 and let $T_i$ be the set of vertices whose $i$th
    bit is 1. Set $\S_i = (S_i,T_i)$.}. For each $i$, let $\W_i$ be a
  minimum $\S_i$-cut.

  Since $r$ is separated from every other vertex in $R$ by at least
  one partition, $\X_r$ is an $(r,R-r)$-cut.  By
  \reflemma{uncrossing-min-cuts}, $\X_r$ contains a minimum
  $(r, R -r)$-cut. For the final property, let $r, q\in R$.  Then at
  least one of the partitions $\S_i$ separates $r$ and $q$. The
  corresponding cut $\W_i$ ensures that either $\X_r \preceq \W_i$ and
  $\X_q \preceq \W_i^T$, or $\X_r \preceq \W_i^T$ and
  $\X_q \preceq \W_i$. Since $\V$ is pairwise disjoint, this implies
  that
  $\X_q \land \X_r \preceq \W_i \land \W_i^T \preceq (\emptyset, R)$.
\end{proof}

\begin{remark}
  The preceding lemma relied on pairwise disjointness of the lattice
  $\V$ and $\R$ as well as symmetry of the bisubmodular function.
  This is necessary for the crucial third property in the lemma which
  is the main reason for improvement in the running time for finding
  isolating cuts.  However, one can obtain the first two properties
  with appropriate modifications even for general lattices and
  bisubmodular functions. At the moment we do not know of any concrete
  algorithmic applications for the general version and hence we do not
  state it explicitly here. As there may be future applications, we
  plan to include the general lemma in a future version of the paper.
\end{remark}

Using the preceding lemma the problem of computing the $f$-minimum
$r$-isolating cuts is reduced to finding such a cut in $\X_r$. The
advantage, in terms of running time, is captured by the disjointness
property: for distinct $r, q \in R$ we have
$\X_r \land \X_q \preceq (\emptyset, R)$.  For each $r$ let
$\X_r = (A_r, B_r)$. Thus we have $\sum_r |A_r| \le |V|$.  Given $r$
and $\X_r$, the problem of computing the $f$-minimum cut
$\Y_r \preceq \X_r$ can in several settings be reduced to solving a
problem that depends only on $|A_r|$ and $|V|$. We capture this in the
following lemma.

\IsolatingCuts*

\subsection{Computing global minimum cuts}

A simple random sampling approach combined with isolating cuts, as
shown in \cite{li-panigrahi} for edge cuts in graphs, yields the following theorem in a
much more abstract setting.

\BisubmodMincut*
\begin{proof}
  Suppose the minimum cut is achieved by $\X = (X,Y) \in \V$, and that
  $\X$ is a cut for a nontrivial partition $\S = (S,T) \in
  \R$. Without loss of generality we assume that
  $\sizeof{S} \leq \sizeof{T}$. Let
  $\ell \in \bracketsof{\sizeof{S}, 2 \sizeof{S}}$; we can guess
  $\ell$ by enumerating all powers of 2 with a $\log{\sizeof{R}}$
  overhead in running time. Let $R' \subseteq R$ sample each $r \in R$
  independently with probability $1/\ell$. With constant probability,
  $\sizeof{R' \cap S} = 1$ and $R' \cap T \neq \emptyset$, in which
  case $\X$ is a minimum $R'$-isolating cut. Thus it suffices to find
  isolating cuts for $R'$ and the total time will have an additional
  $O(\log |R|)$ factor overhead. Now we can apply
  \reflemma{isolating-cuts}.
\end{proof}

We derive the following corollary for symmetric submodular set
functions.

\SymsubmodMincut*
\begin{proof}
  Let $\V$ be the crossing lattice over all bipartitions of $V$
  (\refexample{bipartitions}) and let $\R$ be the crossing lattice
  over all bipartitions of $R$. Consider the bisubmodular function
  $h:\V \rightarrow \reals$ defined from $f$ as in
  \refexample{submod-to-bisubmod} where $h(X,V\setminus X) = f(X)$ for
  each $X \subseteq V$.  Given $(A,R\setminus A) \in \R$ an
  $h$-minimum cut in $\V$ that cuts it corresponds to finding
  $(X,V\setminus X) \in \V$ with minimum $f(X)$ value where
  $A \subseteq X \subseteq V \setminus R$. Such a cut can be found by
  submodular function minimization in $\SFMtime{n}$ time. Thus,
  $\SMtime{n} = O(\SFMtime{n})$. Similarly, given $u \in R$ and
  $A_u \subseteq V$ the problem of finding the $h$-minimum cut $\Y_u$
  such that $\Y_u \preceq (A_u,V\setminus A_u)$ can be reduced to
  submodular function minimization by contracting $V\setminus A_u$
  into a single element and solving submodular function minimization
  on the resulting set which consists of $|A_u|+1$ elements. Note that
  contraction creates a new submodular function; however, the
  evaluation oracle for the new submodular function is essentially the
  same as the one for $f$. Thus, $\SMItime{p,n} = O(\SFMtime{p+1})$.
  By \reftheorem{bisubmod-mincut} one sees that the total time to
  compute the desired minimum cut that separates $R$ with constant
  probability is $O(\SFMtime{n}\log^2 n)$.
\end{proof}


\end{document}